\documentclass[3p]{elsarticle}
\usepackage{multirow}
\usepackage{graphicx,amsmath}
\usepackage{color}
\usepackage{amssymb}
\usepackage{booktabs}
\usepackage{tabularx}
\usepackage{dsfont}
\usepackage{lineno}
\usepackage[ruled,resetcount]{algorithm2e}

\newcommand{\fpt}{{{FPT}}}
\newcommand{\np}{NP}
\newcommand{\nph}{{\np}-hard}

\newcommand{\nphns}{{{\np}-hardness}}
\newcommand{\npc}{{{\np}-complete}}
\newcommand{\npcns}{{{\np}-completeness}}

\newcommand{\wa}{{{W}[1]}}
\newcommand{\wah}{{{W}[1]-hard}}
\newcommand{\wahns}{{{W}[1]-hardness}}

\newcommand{\wb}{{{W}[2]}}
\newcommand{\wbhns}{{{W}[2]-hardness}}

\newcommand{\wbh}{{{W}[2]-hard}}

\newcommand{\xp}{XP}

\newcommand{\wi}{{{W}[$i$]}}
\newcommand{\wih}{{{W}[$i$]-hard}}

\newcommand{\poly}{{{P}}}

\newcommand{\kpeak}{$\mathpzc{k}$}
\newcommand{\mathkpeak}{\mathpzc{k}}

\newcommand{\mathindependentsetdthree}{{\kappa}}

\newcommand{\discandi}{${p}$} 
\newcommand{\mathdiscandi}{{p}}
\newcommand{\mathessize}{{R}}

\newcommand{\appc}{{$r$}} 
\newcommand{\mathappc}{r}

\newcommand{\ccav}{{CCAV}}

\newcommand{\myfig}[1]{Figure~\ref{#1}}
\newcommand{\yes}{Yes}
\newcommand{\no}{No}
\newcommand{\yesins}{{\yes}-instance}
\newcommand{\noins}{{\no}-instance}

\newcommand{\EP}[3]{
\begin{center}
{\small
\begin{tabularx}{0.98\columnwidth}{ll}
\toprule
\multicolumn{2}{c}{\textsc{#1}} \\
\midrule
{\bf Input:}   & \parbox[t]{0.85\columnwidth}{#2\vspace*{1mm}}  \\
{\bf Question:}& \parbox[t]{0.85\columnwidth}{#3\vspace*{.5mm}} \\
\bottomrule
\end{tabularx}
}
\end{center}
}

\newcommand{\onlyfull}[1]{}
\newtheorem{theorem}{Theorem}
\newtheorem{lemma}[theorem]{Lemma}

\newtheorem{observation}{Observation}

\newenvironment{proof}{$\;$\newline \noindent {\sc Proof.}$\;\;\;$\rm}{\qed}

\def\boxit#1{\vbox{\hrule\hbox{\vrule\kern4pt
  \vbox{\kern1pt#1\kern1pt}
\kern2pt\vrule}\hrule}}

\DeclareMathAlphabet{\mathpzc}{OT1}{pzc}{m}{it}

\begin{document}
\begin{frontmatter}
\title{The Control Complexity of {{\appc}}-Approval: from the Single-Peaked Case to the General Case\tnoteref{t1}}
\tnotetext[t1]{A preliminary version of this paper appeared in the proceedings of the international conference AAMAS 2014~\cite{Yangaamas14b}, and informal proceedings of the international workshop M-PREF 2013.}

\author[yyj]{Yongjie Yang
}
\ead{yyongjiecs@gmail.com}
\address[yyj]{
Saarland University, Saarbr\"{u}cken, Germany
}

\author[jg]{Jiong Guo}
\ead{jguo@sdu.edu.cn}
\address[jg]{School of Computer Science and Technology, Shandong University, Jinan, China}

\begin{abstract}
\onlyfull{Many voting problems which are {\npc} in general elections become polynomial-time solvable when restricted to single-peaked elections. }We investigate the complexity of\onlyfull{ constructive control by adding/deleting votes/candidates for} $r$-Approval control problems in {\kpeak}-peaked elections, where at most {\kpeak} peaks are allowed in each vote with respect to an order of the candidates.\onlyfull{ Hence, $1$-peaked elections are exactly single-peaked elections.
In particular, we show that in $2$-peaked elections constructive control by adding votes for $r$-Approval is polynomial-time solvable if $r$ is a constant, but becomes {\npc} if $r$ is part of the input. In {\kpeak}-peaked elections where $\mathkpeak\geq 3$, we prove that the problem is {\nph} for every constant $r\geq 4$. In addition, we prove that constructive control by deleting votes for $r$-approval in {\kpeak}-peaked elections is {\npc} for every constant $r\geq 3$ and $\mathkpeak\geq 2$.} We show that most {\npcns} results in general elections also hold in {\kpeak}-peaked elections even for $\mathkpeak=2,3$. On the other hand, we derive polynomial-time algorithms for some problems for $\mathkpeak=2$. All our {\npcns} results apply to Approval and sincere-strategy preference-based Approval as well. Our study leads to many dichotomy results for the problems considered in this paper,
 with respect to the values of {\kpeak} and $r$. In addition, we study $r$-Approval control problems from the viewpoint of parameterized complexity and achieve both fixed-parameter tractability results and {\wahns} results, with respect to the solution size.\onlyfull{ Concretely, we prove that constructive control by adding/deleting votes for $r$-approval is fixed-parameter tractable for $r$ being a constant, even in general elections. On the other hand, constructive control by deleting candidates for $1$-approval in $3$-peaked elections turns out to be {\wah} with respect to the solution size.}
Along the way exploring the complexity of control problems, we obtain two byproducts which are of independent interest. First, we prove that every graph of maximum degree 3 admits a specific 2-interval representation where every 2-interval corresponding to a vertex contains a trivial interval (a single point) and, moreover, 2-intervals may only intersect at the endpoints of the intervals. Second, we develop a fixed-parameter tractable algorithm for a generalized {\appc}-{\sc{Set Packing}} problem with respect to the solution size, where each element in the given universal set is allowed to occur in more than one  {\appc}-subset in the solution.
\end{abstract}

\end{frontmatter}

\section{Introduction}
Voting is a common method for preference aggregation and collective decision-making, and has applications in many areas such as political elections, multi-agent systems and recommender systems~\cite{Egan2014,DBLP:journals/cj/PittKSA06,DBLP:conf/hci/Popescu13b,DBLP:conf/atal/Xia13}.
However, by Arrow's impossibility theorem~\cite{arrow50} there is no voting system which satisfies a certain set of desirable criteria when
more than two candidates are involved (see~\cite{arrow50} for further details).
One prominent way to bypass Arrow's impossibility theorem is to restrict the domain of preferences, for instance, the single-peaked domain introduced by Black~\cite{Black48}.
Intuitively, in a single-peaked election, one can order the candidates from left to right such that every voter's preference increases first and then decreases after some point as the candidates are considered from left to right.  See {\myfig{fig:sp}} for an example.

\begin{figure}[h!]
\begin{center}
\includegraphics[width=\textwidth]{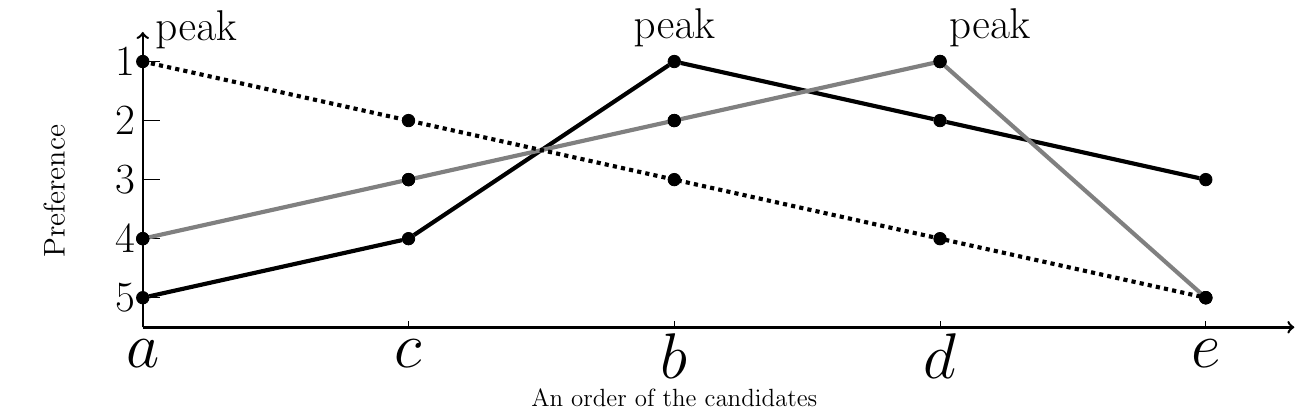}
\caption[A single-peaked election]{An illustration of a single-peaked election. There are five candidates $a,b,c,d,e$ and three voters, whose preferences are illustrated by the dark line, the gray line and the dotted line, respectively. For instance, the preference illustrated by the dark line signifies that $b$ is preferred to $d$ to $e$ to $c$ to $a$.}
\label{fig:sp}
\end{center}
\end{figure}

Recently, the complexity of various voting \mbox{problems} in single-peaked elections has been \mbox{attracting} \mbox{attention} of many researchers both from theoretical computer science and social choice communities~\cite{BrandtBHH2015JAIRbypassingsinglepeakelectionAAAI10,DBLP:conf/ecai/EscoffierLO08,DBLP:journals/iandc/FaliszewskiHHR11,sean09,DBLP:conf/aaai/Walsh07}.
It turned out that many voting problems which are {\npc} in general elections become polynomial-time solvable when restricted to single-peaked elections.
However, most elections in practice are not purely single-peaked, which \mbox{motivates} researchers to study more general models of elections. We refer readers to~\cite{DBLP:journals/mss/BredereckCW16,DBLP:conf/ijcai/CornazGS13,Demangesinglepeaked82,Erdelyi2017,DBLP:journals/ai/FaliszewskiHH14,AAMAS15Yangmanipulationspwidth,Yangaamas14a} for some generalizations and variants of single-peaked elections.
In this paper, we consider a natural generalization of single-peaked elections, called {\kpeak}-peaked elections, where each vote is allowed to have at most {\kpeak} peaks with respect to an order of the candidates.
This generalization might be relevant for many real-world applications.
For example, consider a group of people who are willing to select a special day for an event.
In this setting, each voter may have several special days which he/she prefers for some reason,
and the longer the other days away from these favorite days, the less they are preferred by the voter.
 {\kpeak}-peaked elections with {\kpeak} being a small constant may also arise in the scenario where initially the election is single-peaked and then some voters are bribed and rank some specific candidates higher in order to get some extra benefits (e.g., money, permission, etc.) from the bribers. In addition, {\kpeak}-peaked elections also play an important role in politics~\cite{Cooter2002,Egan2014}. We refer to the work of Egan \cite{Egan2014} for a detailed discussion of how and when {\kpeak}-peaked political elections arise in real-world political settings. Very recently, {\kpeak}-peaked elections have been also studied in the context of facility location problems~\cite{FilosAAAI15LZ}.

In this paper, we are concerned with control problems for {\appc}-Approval restricted to {\kpeak}-peaked elections. In a control problem, there is an external agent (e.g., the chairman in an election) who is willing to influence  the election result by carrying out some strategic behavior. There could be two goals that the external agent wants to reach. One goal is to make some distinguished candidate win the election, and the other goal is to make the distinguished candidate lose the election. A control problem with the former goal is referred to as a {\it{constructive control}} and with the latter goal as a {\it{destructive control}} in the literature~\cite{Bartholdi92howhard,DBLP:journals/ai/HemaspaandraHR07}. Moreover, the strategic behavior may involve adding/deleting a limited number of votes/candidates. We refer readers to~\cite{Bartholdi92howhard,DBLP:journals/mlq/ErdelyiNR09,DBLP:journals/ai/HemaspaandraHR07} for further discussions on control problems. In this paper, we study only constructive control. Hereinafter, ``control'' means ``constructive control''.

Approval is one of the most famous voting {systems} and has been extensively studied both in theory and in practice.
In {\it{Approval}}, we are {given} a set $\mathcal{C}$ of {\textit{candidates}} and a set $\mathcal{V}$ of {\textit{voters}},
each of whom {approves} or disapproves each candidate $c\in\mathcal{C}$. The candidates with the most approvals are {\it{winners}}.
{\it{{\appc}-Approval}} is a variant of {Approval} where {\appc} is a positive integer. In particular, in {\appc}-approval each voter $v$ casts a vote $\pi_v$ defined as a {linear} order
over the candidates and approves exactly the top-{\appc} candidates. 1-Approval is often referred to as {\it{Plurality}} in the literature~\cite{AAMAS15AzizGGMMW,Bartholdi92howhard}. Another {prominent} variant of Approval is the {\it{sincere-strategy preference-based Approval}}~(SP-AV for short), proposed by Brams and Sanver~\cite{Brams76}.
In SP-AV, each voter provides both a linear order of the candidates and a subset $C$ of candidates such that the candidates are approved according to $C$, and the ``admissible'' and ``sincere'' properties should be fulfilled. In particular, for each voter $v$, the set $C$ of approved candidates should contain the top-$r_v$ candidates (according to the linear order) for some integer $r_v$. We refer to~\cite{Brams76,DBLP:journals/mlq/ErdelyiNR09} for the precise definition of SP-AV.

\begin{table}[h!]
\begin{center}
\renewcommand\arraystretch{1.5}
\begin{tabular}{|c|c|c|c|c|} \hline
 & \multicolumn{4}{c|}{{\kpeak}-peaked elections} \\ \cline{2-5}

 & $\mathkpeak=1$ & $\mathkpeak=2$ & $\mathkpeak\geq 3$& $\mathkpeak=\lceil m/2 \rceil$ \\ \hline\hline

 & \multirow{4}{*}{\poly $(\spadesuit)$} & {\appc} is a constant: \hspace{0mm}& $\mathappc \leq 3$: {\poly} $(\diamondsuit)$ & \hspace*{-1mm} $\mathappc \leq 3$: {\poly} $(\diamondsuit)$ \hspace*{0em}\\

adding& & \hspace*{0mm} {\bf{\poly}} (Theorem~\ref{3app}) & $\mathappc \geq 4$ {\bf{\npc}}:&\hspace*{-4mm} $\mathappc \geq 4$: {\npc} $({\diamondsuit})$\\

votes & &{\appc} is not a constant: \hspace{0mm} & (Theorem~\ref{theorem:4acav3p})& $r$ is a constant: {\bf{\fpt}}\\

 & &\hspace*{0mm}{{\bf{\npc}} (Theorem~\ref{rav2nphard})}&& w.r.t. $R$ (Theorem~\ref{fptacav})\\ \hline

 &  \multirow{3}{*}{\poly $({\spadesuit})$}  & \multicolumn{2}{c|}{}& $\mathappc\leq 2$: {\poly} $(\diamondsuit)$ \hspace*{0em}\\

deleting& &\multicolumn{2}{c|}{$\mathappc \leq 2$: {\poly} $(\diamondsuit)$} &$\mathappc \geq 3$: {\npc} $({\diamondsuit})$\\

votes & &\multicolumn{2}{c|}{$\mathappc \geq 3$: {{\bf{\npc}}} (Theorem~\ref{theorem:3ad2}) } & $r$ is a constant: {\bf{\fpt}}\\

 &&\multicolumn{2}{c|}{}& w.r.t. $R$ (Theorem~\ref{theorem:3ad2fpt})\\ \hline

adding& \multirow{2}{*}{{\poly} $(\spadesuit)$} & \multirow{2}{*}{$\mathappc\geq 1$: {{\npc}} $({\triangle})$} &\multirow{2}{*}{\hspace*{0mm} $\mathappc\geq 1$: {\npc} $({\triangle})$} & $\mathappc\geq 1$: {\wbh}\\

candidates&&&& w.r.t. $R$ $(\clubsuit)$\\ \hline

deleting & \multirow{2}{*}{{\poly} $(\spadesuit)$} & \multirow{2}{*}{$\mathappc\geq 1$: {{\npc}} $({\triangle})$} &$\mathappc\geq 1$: {\bf{\wah}}& $\mathappc\geq 1$: {\wbh}\\

candidates&&& w.r.t. $R$ (Theorem~\ref{theorem:deletingcandidates})& w.r.t. $R$ $(\blacklozenge)$\\ \hline
\end{tabular}
\caption{A summary of the complexity of {\appc}-Approval control problems.
Our results are in bold. In this table, ``{\poly}" stands for ``polynomial-time solvable'', and $R$ in an entry is the solution size in the corresponding problem. Note that $\lceil m/2\rceil$-peaked elections are general elections, where $m$ denotes the number of candidates.
All results apply to the unique-winner and the nonunique-winner models.
Moreover, all our {\npc} results apply to  both Approval and SP-AV. However, there is no ``{\appc}'' in either case.
Results marked by $\diamondsuit$ are from \protect\cite{DBLP:conf/icaart/Lin11}, by $\clubsuit$ from \protect\cite{DBLP:journals/tcs/LiuFZL09}, by $\spadesuit$ from \protect\cite{DBLP:journals/iandc/FaliszewskiHHR11}, by $\triangle$ from \protect\cite{DBLP:journals/ai/FaliszewskiHH14} and by $\blacklozenge$ from \protect\cite{DBLP:journals/tcs/BetzlerU09}.}
\label{tableresults}
\end{center}
\end{table}

Hemaspaandra, Hemaspaandra and Rothe~\cite{DBLP:journals/ai/HemaspaandraHR07} proved that control by adding/deleting votes for Approval is {\npc}.
The proofs can be adapted to show the {\np}-completeness of control by adding/deleting votes for SP-AV~\cite{DBLP:journals/mlq/ErdelyiNR09}.
Lin~\cite{DBLP:conf/icaart/Lin11} established many dichotomy results for $r$-Approval control problems with respect to the values of $r$. In particular, Lin proved that control by adding votes for $r$-Approval is {\npc} if and only if $r\geq 4$, and control by deleting votes for $r$-Approval is {\np}-complete if and only if $r\geq 3$.
As for control by modification of candidates, Approval turned out to be immune\footnote{A voting system is immune to a control problem if one cannot make a non-winning candidate a winner by performing the corresponding strategic behavior.} to control by adding candidates and polynomial-time solvable for control by deleting candidates~\cite{DBLP:journals/ai/HemaspaandraHR07}. However, control by adding/deleting candidates for {\appc}-Approval is {\npc}, even when degenerated to 1-Approval~\cite{Bartholdi92howhard}. The {\npcns} also holds for SP-AV~\cite{DBLP:journals/mlq/ErdelyiNR09}. 
Recently, Approval and {\appc}-Approval control problems have been considered in single-peaked {elections}. In particular, Faliszewski et al.~\cite{DBLP:journals/iandc/FaliszewskiHHR11} proved that control by adding/deleting votes for Approval is polynomial-time solvable in single-peaked elections\footnote{In~\cite{DBLP:journals/iandc/FaliszewskiHHR11}, an Approval election is single-peaked if there is an order of the candidates such that each voter's approved candidates are contiguous in the order.}. Moreover, control by adding/deleting candidates for 1-Approval is polynomial-time solvable in single-peaked elections~\cite{DBLP:journals/iandc/FaliszewskiHHR11}.

Motivated by the {\npcns} in the general case and the polynomial-time solvability in the single-peaked case, we study the complexity of control by adding/deleting votes/candidates for {\appc}-Approval in {\kpeak}-peaked elections, aiming at exploring the complexity border of these control problems with respect to various values of {\kpeak} and $r$. Faliszewski, Hemaspaandra and Hemaspaandra~\cite{DBLP:journals/ai/FaliszewskiHH14} studied a notion of nearly single-peaked elections which is called Swoon-SP and is a special case of 2-peaked elections. They proved that control by adding/deleting candidates for 1-Approval is {\npc} when restricted to Swoon-SP elections, implying the {\npcns} of these problems in 2-peaked elections. We complement their results by studying the adding/deleting votes counterpart. Our findings are summarized in Table~\ref{tableresults}. In particular, we show that if $r$ is a constant,  control by adding votes for {\appc}-Approval is polynomial-time solvable in 2-peaked elections, but becomes {\npc} in 3-peaked elections for every constant $r\geq 4$. If {\appc} is not a constant, we show that control by adding votes for {\appc}-Approval already becomes {\np}-complete in 2-peaked elections. In addition,  we prove that control by deleting votes for $r$-Approval is {\npc} in 2-peaked elections, even for every constant $r\geq 3$.

Apart from the above results, we present some results for {\appc}-Approval control problems with respect to their parameterized complexity.
Recently, the parameterized complexity of various voting problems has received a considerable amount of attention,
see, e.g.,~\cite{DBLP:journals/jair/BredereckFNT16,
DBLP:journals/tcs/DeyMN16,
DBLP:journals/jcss/ErdelyiFRS15,
DBLP:journals/tcs/LiuFZL09,
DBLP:journals/ipl/LiuZ10,
AAMAS15MisraNS,
DBLP:journals/tcs/SkowronYFE15,
DBLP:conf/ecai/Yang14,
AAMAS15Yangmanipulationspwidth,
DBLP:journals/jco/YangG17}.
A {\it{parameterized problem}} is a language $L \subseteq \Sigma^*\times\mathbb{N}$, where $\Sigma$ is a finite alphabet.
The first component is called the {\it{main part}} and the second component is called the {\it{parameter}} of the problem. Downey and Fellows~\cite{fellows99} established the parameterized complexity hierarchy: \[\text{\fpt}\subseteq \text{\wa}\subseteq \text{\wb}\subseteq \cdots \subseteq \text{\xp},\] where the class {\fpt} (stands for ``fixed-parameter tractable'') includes all parameterized problems which admit $O(f(\kappa)\cdot|I|^{O(1)})$-time algorithms, where $I$ is the main part, $|I|$ is the size of $I$, $\kappa$ is the parameter and $f$ can be any computable function. Such algorithms are called {\it{{\fpt}-algorithms}}. For a positive integer $i$,  a parameterized problem is {\it{\wih}} if all problems in {\wi} are {\fpt}-reducible to the problem. It is widely believed that {\wih} problems where $i\geq 1$ do not admit {\fpt}-algorithms (otherwise the parameterized complexity hierarchy collapses). Hence, {\wih} problems are referred to as fixed-parameter intractable problems in the literature.

Given two parameterized problems $Q$ and $Q'$, an {\it{{\fpt}-reduction}} from $Q$ to $Q'$ is an
algorithm that takes as input an instance $(I,\kappa)$ of $Q$ and outputs an instance $(I',\kappa')$ of $Q'$ such that
\begin{enumerate}
\item the algorithm runs in $f(\kappa)\cdot |I|^{O(1)}$ time, where $f$ is a computable function;
\item $(I,\kappa)\in Q$ if and only if $(I',\kappa')\in Q'$; and
\item $\kappa'\leq g(\kappa)$, where $g$ is a computable function.
\end{enumerate}

Liu et al.~\cite{DBLP:journals/tcs/LiuFZL09} proved that control by adding votes for Approval is {\wah} and control by deleting votes for Approval is {\wbh}, with the number of added and deleted votes as parameters, respectively.
In addition, Liu et al.~\cite{DBLP:journals/tcs/LiuFZL09} proved that control by adding candidates for 1-Approval is {\wbh}, with the number of added candidates as the parameter. Betzler and Uhlmann~\cite{DBLP:journals/tcs/BetzlerU09} complemented these results by proving that control by deleting candidates for 1-Approval is {\wbh}, with the number of deleted candidates as the parameter. The {\wbhns} reductions of control by adding/deleting candidates for 1-Approval in~\cite{DBLP:journals/tcs/LiuFZL09} and~\cite{DBLP:journals/tcs/BetzlerU09} can be extended to $r$-Approval for every constant $r\geq 2$. In this paper, we prove that control by deleting candidates for 1-Approval restricted to 3-peaked elections is {\wah} with the number of deleted candidates as the parameter. Regarding general elections (the domain of preferences is not restricted), we prove that for constant $r$, control by adding/deleting votes for $r$-Approval is {\fpt} with respect to the number of added/deleted votes.

Along the way establishing our results described above, we obtain two byproducts which are of independent interest. First, to show the {\npcns} of control by deleting votes for $r$-Approval in 2-peaked elections, we study a property of graphs of maximum degree 3 (Lemma~\ref{lemma:ids}). By and large, we prove that every graph of maximum degree 3 has a 2-interval representation (a 2-interval is a set consisting of two intervals over the real line) where every 2-interval contains one trivial interval and, moreover, every two 2-intervals may only intersect at their endpoints. Since many graph problems, such as the {\sc{Vertex Cover}} problem, the {\sc{Hamiltonian Cycle}} problem and the {\sc{Dominating Set}} problem, remain {\nph} when restricted to graphs of maximum degree 3~\cite{garey,DBLP:journals/tcs/GareyJS76,DBLP:journals/tcs/Picouleau94}, this property may be useful in developing algorithms for these problems, or proving {\nphns} of further problems restricted to graphs of maximum degree 3. Second, to show the fixed-parameter tractability of control by adding votes for $r$-Approval in general elections, we derive an {\fpt}-algorithm for a generalization of the {\nph} problem {\sc{{$r$}-Set Packing}}, in which given a universal set $X$ and a collection $U$ of {$r$}-subsets of $X$ where $r$ is a constant, one asks for a subcollection of size {\kpeak} such that no two $r$-subsets in the subcollection intersect. In the generalization, we allow each element $x\in X$ to occur in at most $f(x)$ many $r$-subsets of the desired subcollection, where $f$ is a given mapping from $X$ to positive integers. Thus, if $f(x)=1$ for every $x\in X$, we have the $r$-{\sc{Set Packing}} problem. Jia, Zhang and Chen~\cite{DBLP:journals/jal/JiaZC04setpacking} crafted an ingenious {\fpt}-algorithm for the $r$-{\sc{Set Packing}} problem, with respect to {\kpeak}.
Based on this {\fpt}-algorithm, we devise an {\fpt}-algorithm for the generalized $r$-{\sc{Set Packing}} problem, with respect to {\kpeak}.

\section{PRELIMINARIES}
We will need the following notations. Unless stated otherwise, all numerical data are integers.

{\bf{Multisets.}} A {\it{multiset}} $S=\{s_1, s_2, \dots , s_{\scriptsize{|S|}}\}$ is a generalization of a set where {\it{objects}} are allowed to appear more than once, that is, $s_i=s_j$ is allowed for $i\neq j$. An {\it{element}} of $S$ is one copy of some object. We use $s\in_{+} S$ to denote that $s$ is an element of $S$. The {\it{size}} of $S$, denoted by $|S|$, is the number of elements in $S$. For instance, the size of the multiset $\{1,1,2,2,2,2,4\}$ is $7$. For two multisets $A$ and $B$, we use $A\uplus B$ to denote the multiset containing all elements from $A$ and $B$, 
and use $A\ominus B$ to denote the multiset containing, for each object $s$ of $A$, $\max\{0, n_1-n_2\}$ copies of $s$, where $n_1$ and $n_2$ denote the numbers of copies of $s$ in $A$ and $B$, respectively. For example, for $A=\{1,1,1,2,3,3,4\}$ and $B=\{1,2,3\}$, $A\uplus B=\{1,1,1,1,2,2,3,3,3,4\}$ and $A\ominus B=\{1,1,3,4\}$. A multiset $B$ is a {\it{submultiset}} of a multiset $A$ if for every object $s$ that occurs $n$ times in $B$, $A$ contains at least $n$ copies of $s$. We use $B\sqsubseteq A$ to denote that $B$ is a submultiset of $A$.
\medskip

{\bf{{\appc}-Approval.}} Let $\mathcal{C}=\{c_1,c_2, \dots ,c_m\}$ be a set of $m$ candidates and $\mathcal{V}$ a set of voters where every $v\in \mathcal{V}$ casts a vote $\pi_v$ defined as a linear order over $\mathcal{C}$. For a vote $\pi_v$ defined as a linear order $(c_{\beta(1)}, \dots ,c_{\beta(m)})$, $\pi_v(c_{\beta(i)})$ denotes the {\it{position}} of the candidate $c_{\beta(i)}$ in $\pi_v$, i.e., $\pi_v(c_{\beta(i)})=i$, where $\beta$ is a permutation of $\{1,2, \dots ,m\}$, i.e., $\{\beta(1), \dots ,\beta(m)\}=\{1, \dots ,m\}$. The multiset of votes cast by voters in $\mathcal{V}$ is denoted by $\Pi_{\mathcal{V}}$. The tuple $(\mathcal{C},\Pi_{\mathcal{V}})$ is called an {\it{election}}. In $r$-Approval, each voter $v$ gives to each candidate $c$ one point if $\pi_v(c)\leq \mathappc$, and zero points otherwise.  For a vote $\pi_v$, let $1(v)$ denote the set of candidates who get one point and $0(v)$ the set of candidates who get zero points from $\pi_v$, i.e., $1(v)=\{c\in \mathcal{C}\mid \pi_v(c)\leq \mathappc\}$ and $0(v)=\mathcal{C}\setminus 1(v)$. For a candidate $c$, let $SC_\mathcal{V}(c)$ be the total score of $c$ from $\Pi_{\mathcal{V}}$, i.e., $SC_{\mathcal{V}}(c)=|\{\pi_v\in_+ \Pi_{\mathcal{V}}\mid c\in 1(v)\}|$. Candidates with the highest total score are  called the {\it{winners}} of $(\mathcal{C}, \Pi_{\mathcal{V}})$ with respect to $r$-Approval. If there is only one winner, we call it the {\it{unique winner}}; otherwise, we call them {\it{cowinners}}.

For a vote $\pi_v$ and a subset $C\subseteq \mathcal{C}$, let $\pi_v(C)$
be the {\textit{partial vote}} of $\pi_v$ restricted to $C$ such that in $\pi_v(C)$ the relative order between every two distinct candidates in $C$ preserves the same as
in $\pi_v$. For example, for $\pi_v=(a,b,c,d,e)$, we have $\pi_v(\{b,d,e\})=(b,d,e)$. For a multiset $\Pi$ of votes and a subset $C\subseteq \mathcal{C}$, let $\Pi(C)$ be the multiset obtained from $\Pi$ by replacing each $\pi\in_+ \Pi$ by  $\pi(C)$.
\medskip

{\bf{Single-peaked/{\kpeak}-peaked elections.}} An election~$(\mathcal{C}, \Pi_{\mathcal{V}})$ is {\textit{single-peaked}} if there is an order
$\mathcal{L}$ of $\mathcal{C}$ such that for every vote $\pi_v\in_+ \Pi_{\mathcal{V}}$ and every three candidates $a, b, c \in \mathcal{C}$ with
$a \; \mathcal{L} \; b \; \mathcal{L} \; c$ or $c \; \mathcal{L} \; b \; \mathcal{L} \; a$, $\pi_v(c)< \pi_v(b)$ implies $\pi_v(b)<\pi_v(a)$, where $a\; \mathcal{L}\; b$ means that $a$ is ordered on the left-side of $b$ in $\mathcal{L}$.
The candidate $c$ such that $\pi_v(c)=1$ is the {\it{peak}} of $\pi_v$ with respect to $\mathcal{L}$.

\begin{figure}[h!]
  \begin{center}
    \includegraphics[width=\textwidth]{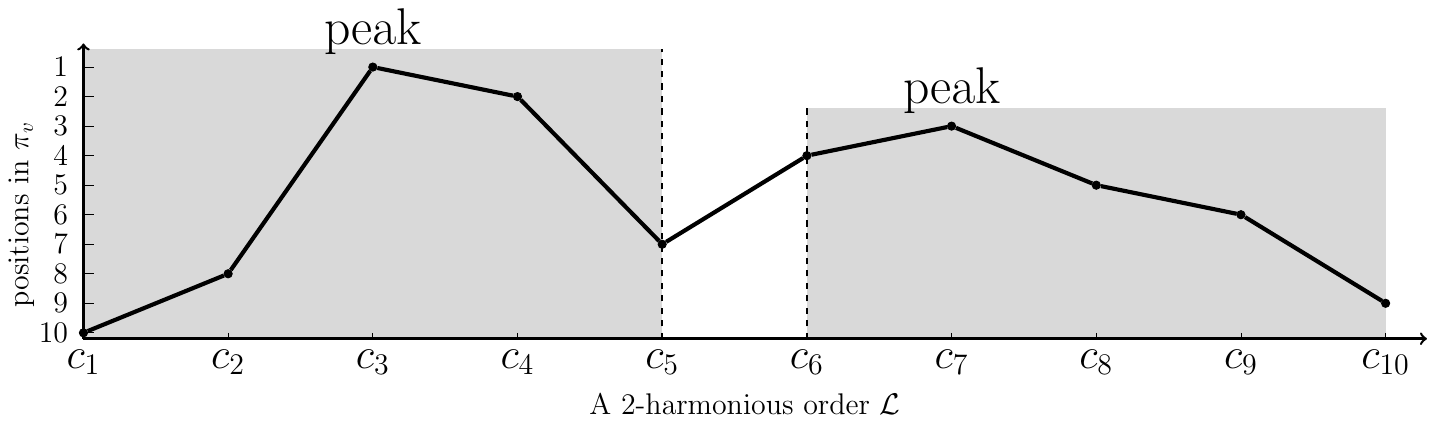}
    \end{center}
    \caption[A 2-peaked vote]{This figure shows a 2-peaked vote ${\pi_v=(c_3, c_4,c_7, c_6, c_8,c_9,c_5,c_2, c_{10}, c_1)}$ with respect to the 2-harmonious order ${\mathcal{L}=(c_1, c_2, \dots, c_{10})}$. Here,  ${\mathcal{L}}$ is {partitioned} into ${{L}_1=(c_1, c_2, c_3, c_4, c_5)}$ and ${{L}_2=(c_6, c_7, c_8, c_9, c_{10})}$.
    }
    \label{figpeakk}
\end{figure}

For an order $\mathcal{L}=(c_{1}, c_{2},\dots, c_{m})$ of $\mathcal{C}$ and a vote $\pi_v$, we say that $\pi_v$ is {\textit{{\kpeak}-peaked}} with respect to $\mathcal{L}$
if there is a ${\mathkpeak}'$-partition $L_1=(c_{1}, \dots ,c_{i_1}), L_2=(c_{i_1+1}, \dots ,c_{i_2}),\dots, L_{{\mathkpeak}'}=(c_{i_{\mathkpeak'-1}+1}, \dots ,c_{m})$ of $\mathcal{L}$ such that ${\mathkpeak}'\leq {\mathkpeak}$ and $\pi_v(\mathcal{C}(L_x))$ is single-peaked with respect to $L_x$ for all $1\leq x\leq {\mathkpeak}'$, where $\mathcal{C}(L_x)$ is the set of candidates appearing in $L_x$.
See Figure~\ref{figpeakk} 
for an example.

An election is {\it{{\kpeak}-peaked}} if there is an order $\mathcal{L}$ of $\mathcal{C}$ such that every vote in the election is {\kpeak}-peaked with respect to $\mathcal{L}$.
Here $\mathcal{L}$ is called a {\it {\kpeak}-harmonious order} of the election. Thus, 1-peaked elections are exactly single-peaked elections. Moreover, every election of $m$ candidates is an
$\lceil m/2\rceil$-peaked election.
\medskip

{\bf{Problem definitions.}} The problems  studied in this paper are defined as follows. Throughout this paper, let $p$ denote the distinguished candidate.

\EP
{
{\appc}-Approval Control by Adding Votes in {\kpeak}-Peaked Elections ({\appc}-AV-{\kpeak})
}
{
An election $(\mathcal{C}, \Pi_{\mathcal{V}})$, a distinguished candidate $p\in \mathcal{C}$, a multiset $\Pi_{\mathcal{T}}$ of unregistered votes, a {\kpeak}-harmonious order $\mathcal{L}$ such that all votes in $\Pi_{\mathcal{V}}$ and all votes in $\Pi_{\mathcal{T}}$ are {\kpeak}-peaked with respect to $\mathcal{L}$, and an integer $0\leq R\leq |\Pi_{\mathcal{T}}|$.
}
{
Is there a $\Pi_{\mathcal{T}'}\sqsubseteq \Pi_{\mathcal{T}}$ such that $|\Pi_{\mathcal{T}'}|\leq R$ and $p$ wins $(\mathcal{C},\Pi_{\mathcal{V}}\uplus \Pi_{\mathcal{T}'})$ with respect to {\appc}-Approval?
}

\EP{
{\appc}-Approval Control by Deleting Votes in {\kpeak}-Peaked Elections ({\appc}-DV-{\kpeak})
}
{An election $(\mathcal{C},\Pi_{\mathcal{V}})$, a distinguished candidate $p\in \mathcal{C}$, a
{\kpeak}-harmonious order $\mathcal{L}$ such that all votes in $\Pi_{\mathcal{V}}$ are {\kpeak}-peaked with respect to $\mathcal{L}$, and an integer $0\leq R\leq |\Pi_{\mathcal{V}}|$.
}
{Is there a $\Pi_{\mathcal{T}}\sqsubseteq \Pi_{\mathcal{V}}$ such that $|\Pi_{\mathcal{T}}|\leq R$ and $p$ wins $(\mathcal{C},\Pi_{\mathcal{V}}\ominus \Pi_{\mathcal{T}}$) with respect to {\appc}-Approval?
}

\EP{
{\appc}-Approval Control by Deleting Candidates in {\kpeak}-Peaked Elections ({\appc}-DC-{\kpeak})}
{An election $(\mathcal{C},\Pi_{\mathcal{V}})$, a distinguished candidate $p\in \mathcal{C}$, a
{\kpeak}-harmonious order $\mathcal{L}$ such that all votes in $\Pi_{\mathcal{V}}$ are {\kpeak}-peaked with respect to $\mathcal{L}$, and an integer $0\leq R\leq |\mathcal{C}|-1$.
}
{Is there a $C\subseteq \mathcal{C}\setminus\{p\}$ such that $|C|\leq R$ and $p$ wins  $(\mathcal{C}\setminus C,\Pi_{\mathcal{V}}(\mathcal{C}\setminus C))$ with respect to {\appc}-Approval?
}

We use {\appc}-AV, {\appc}-DV and {\appc}-DC to denote the above problems without the {\kpeak}-peakedness restriction (i.e., there is no {\kpeak}-harmonious order in the input), respectively. Following the convention, for each of the above problems, we distinguish between the {\it{unique-winner model}} and the {\it{nonunique-winner model}}, where in the unique-winner model winning an election means to be the unique winner, while in the nonunique-winner model winning an election means to be a winner, i.e., either the unique winner or a cowinner.
\bigskip

{\bf{Remarks.}} All our results apply to both the unique-winner model and the nonunique-winner model.
For the sake of clarity, our proofs are solely based on the unique-winner model. However, the result for the nonunique-winner model of a problem can be obtained by slightly modifying the proof for the unique-winner model of the same problem. All our {\npcns} results work for Approval and SP-AV as well.

\section{2-Peaked Elections}
In this section, we study {\appc}-Approval control problems restricted to 2-peaked elections. The following three theorems summarize our findings.

\begin{theorem}
\label{3app}
{\appc}-AV-2 is polynomial-time solvable for every constant {\appc}.
\end{theorem}

Recall that {\appc}-AV is {\npc} for every constant \mbox{$r\geq 4$} but polynomial-time solvable when restricted to single-peaked elections~\cite{DBLP:journals/iandc/FaliszewskiHHR11}. Theorem~\ref{3app} shows that the polynomial-time solvability of {\appc}-AV remains when extending from single-peaked elections to 2-peaked elections, for {\appc} being a constant. This bound is tight as indicated by the following theorem. More precisely, if {\appc} is not a constant, {\appc}-AV becomes {\npc} in 2-peaked elections, in contrast to the polynomial-time solvability of the problem in single-peaked elections~\cite{DBLP:journals/iandc/FaliszewskiHHR11}.

\begin{theorem}\label{rav2nphard}
\label{uacav}
{\appc}-AV-2 is {\npc} if {\appc} is a part of the input.
\end{theorem}

Faliszewski et al.~\cite{DBLP:journals/iandc/FaliszewskiHHR11} proved that {\appc}-DV-1 is polynomial-time solvable, even when {\appc} is not a constant. The following theorem shows that by increasing the number of peaks only by one, this problem becomes {\npc} for every constant $r\geq 3$. Note that $r$-DV is polynomial-time solvable for $r\leq 2$~\cite{DBLP:conf/icaart/Lin11}.

\begin{theorem}\label{theorem:3ad2}
{\appc}-DV-2 is {\npc} for every constant $r\geq 3$.
\end{theorem}

\subsection{Proof of Theorem~{\ref{3app}}}
We prove Theorem~\ref{3app} by deriving a polynomial-time algorithm for $r$-AV-2 based on dynamic programming. Recall that control by adding votes for {\appc}-Approval is polynomial-time solvable for $\mathappc\leq 3$ even in general elections. Hence, we need only to prove the theorem for every constant $\mathappc\geq 4$. Let $((\mathcal{C}, \Pi_{\mathcal{V}}), p \in \mathcal{C}, \Pi_{\mathcal{T}}, \mathcal{L}, R)$ be an instance of {\appc}-AV-2 where $\mathappc\geq 4$.
For $c\in \mathcal{C}$, let $\overleftarrow{c}(1)$ be the candidate lying immediately before $c$ in $\mathcal{L}$
and $\overleftarrow{c}(i)$ be the candidate lying immediately before $\overleftarrow{c}(i-1)$ in $\mathcal{L}$. Similarly, we
use $\overrightarrow{c}(1)$ and $\overrightarrow{c}(i)$ to denote the candidates lying immediately after $c$ and $\overrightarrow{c}(i-1)$, respectively.
For example, if $\mathcal{L}=(a, b, c, d, e, f, g, h)$,
then $\overrightarrow{d}(1)=e, \overrightarrow{d}(4)=h, \overleftarrow{d}(1)=c$ and $\overleftarrow{d}(3)=a$.

Given an order $A=(a_{1}, a_{2},\dots, a_{n})$,
a {\it discrete interval} $I$ over $A$ is a sub-order $(a_i,\dots, a_{i+t})$ of $A$, where $0\leq t\leq n-1$ and $1\leq i\leq n-t$.
We denote the left-most element $a_i$ by $l(I)$ and the right-most element $a_{i+t}$ by $r(I)$.
We also use $A(l(I),r(I))$ to denote $I$.
Let $\mathcal{S}(I)$ denote the set of elements appearing in $I$ and, for notational simplicity, let $|I|=|\mathcal{S}(I)|$ be the size of $I$.
For example, for a discrete interval $I=A(3,6)$ over the order $A=(2, 5, 3, 10, 4, 6, 0)$,
$\mathcal{S}(I)$ is $\{3, 4, 6, 10\}$.
A {\it{$k$}-discrete interval} over an order $A$ is a collection of $k$ disjoint discrete intervals over $A$,
where ``disjoint'' means that no element in $A$ appears in more than one of these discrete intervals.
For a $k$-discrete interval $\mathcal{I}$, let $\mathcal{S}(\mathcal{I})=\bigcup_{I\in\mathcal{I}}{\mathcal{S}(I)}$.
The following observations are useful.

\begin{observation}\label{obinter}
For each {\kpeak}-peaked election $(\mathcal{C},\Pi_{\mathcal{V}})$ associated with a {\kpeak}-harmonious order
$\mathcal{L}$ over $\mathcal{C}$, and each vote $\pi_v\in_+ \Pi_{\mathcal{V}}$, there is a $\mathkpeak'$-discrete
interval $\mathcal{I}$ over $\mathcal{L}$ such that $0<\mathkpeak'\leq \mathkpeak$ and $1(v)=\mathcal{S}(\mathcal{I})$.
\end{observation}
\vspace*{-3pt}

By Observation~\ref{obinter}, for every 2-peaked vote $\pi_v$ with respect to a 2-harmonious order $\mathcal{L}$, $1(v)$ can be represented by a $2$-discrete interval or a $1$-discrete interval over $\mathcal{L}$.
See~\myfig{fig:2pinterval} for an example.

\begin{figure}
\begin{center}
\includegraphics[width=\textwidth]{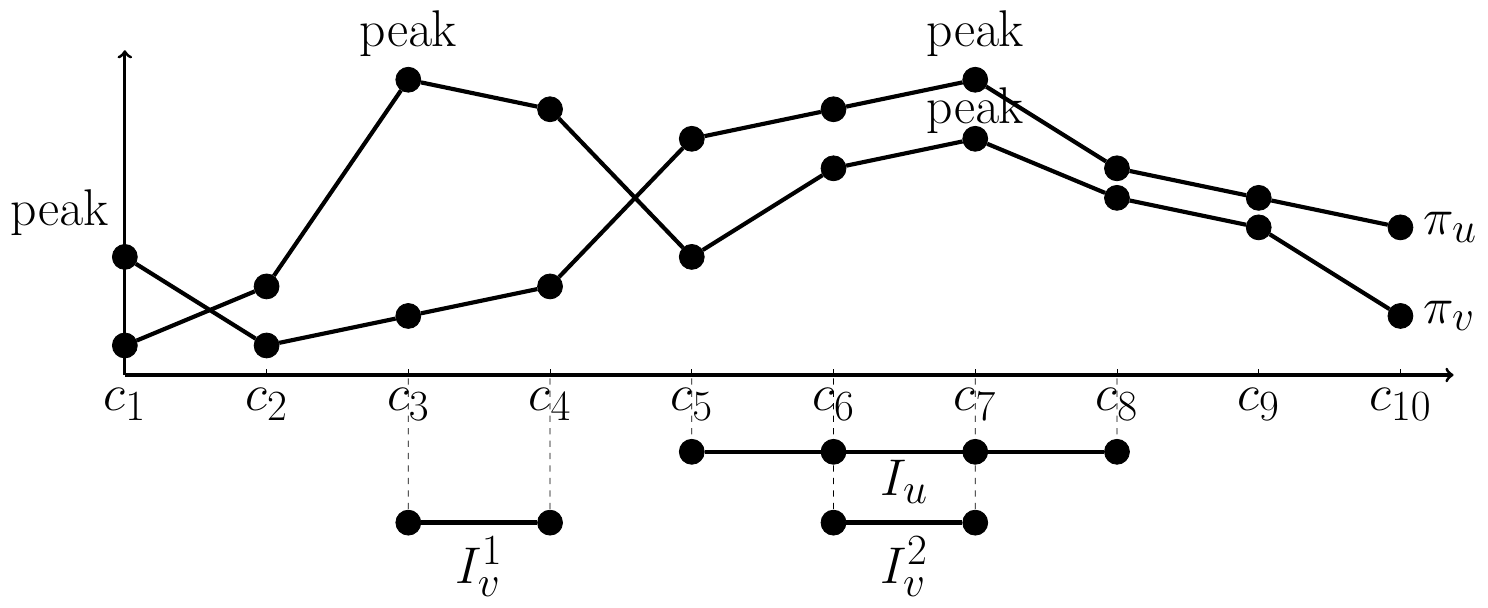}
\end{center}
\caption[Two 2-peaked votes represented by 2-discrete intervals]{This figure shows two votes ${\pi_v=(c_3, c_4, c_7, c_6, c_8, c_9, c_5, c_2, c_{10}, c_1)}$ and ${\pi_u=(c_7,c_6,c_5, c_8, c_9,c_{10}, c_1, c_4,c_3, c_2)}$. Each vote gives one point to its top-4 candidates.
${1(v)}$ is represented by a 2-discrete interval ${\{I_v^1=(c_3,c_4),I_v^2=(c_6,c_7)\}}$ and ${1(u)}$ is represented by a 1-discrete interval ${\{I_u=(c_5,c_6,c_7,c_8)\}}$.
}
\label{fig:2pinterval}
\end{figure}

\begin{observation}\label{obs3}
Every {\yesins} of {\appc}-AV has a solution where each vote approves
$p$.
\end{observation}

We first derive a polynomial-time algorithm for 4-AV-2.
It is easy to generalize the algorithm to {\appc}-AV-2 for every constant ${\mathappc}\geq 5$. Due to Observation~\ref{obs3}, we can safely assume that $p\in 1(v)$ for each $\pi_v\in_+ \Pi_{\mathcal{T}}$.
Due to Observation~\ref{obinter}, for every vote $\pi_v\in_+ \Pi_{\mathcal{T}}$, $1(v)$ can be represented by a 2-discrete interval $\mathcal{I}_v=\{I_v^{\overline{p}}, I_v^p\}$, where $I_v^p$ is the discrete interval including $p$ and $I_v^{\bar{p}}$ is the discrete interval without $p$ in it, or
a 1-discrete interval $\mathcal{I}_v=\{I_v^p\}$ where $p\in \mathcal{S}(I_v^p)$.
Let $\Pi$ be the multiset of all votes $\pi_v\in_+\Pi_{\mathcal{T}}$ where $1(v)$ is represented by a 1-discrete interval over $\mathcal{L}$. We say that two votes have the same {\it type} if they approve the same candidates. Since every voter approves exactly 4 candidates, $\Pi$ has at most four different types of votes:

(1) votes approving $\overleftarrow{p}(3), \overleftarrow{p}(2),\overleftarrow{p}(1), p$;

(2) votes approving $\overleftarrow{p}(2), \overleftarrow{p}(1), p,\overrightarrow{p}(1)$;

(3) votes approving $\overleftarrow{p}(1), p, \overrightarrow{p}(1),\overrightarrow{p}(2)$; and

(4) votes approving $p, \overrightarrow{p}(1), \overrightarrow{p}(2),\overrightarrow{p}(3)$.

Then, we enumerate all possibilities of how many votes in a potential solution are from each of the four types of votes in $\Pi$. For each possibility, we recalculate the scores of the candidates by incorporating the corresponding votes into the election, and update $R$ accordingly. We immediately discard all possibilities leading to a negative value of $R$. In addition, after recalculating the scores we remove all unregistered votes which can be represented by 1-discrete intervals. This breaks down the original instance into at most $R^4$ subinstances. It is clear that the original instance is a {\yesins} if and only if at least one of the subinstances is a {\yesins}. In the following, we show how to solve each subinstance in polynomial time. Let $\Pi_{\mathcal{T}}$ be the multiset of unregistered votes in the subinstance currently considered. As discussed above, each vote in $\Pi_{\mathcal{T}}$ is represented by a 2-discrete interval. Let
$\vec{\Pi}_{\mathcal{T}}=(\pi_{v_1}, \pi_{v_2}, \dots ,\pi_{v_{|\mathcal{T}|}})$ be an order of $\Pi_{\mathcal{T}}$ such that
$r(I_{v_i}^{\overline{p}}) = r(I_{v_j}^{\overline{p}})$ or $r(I_{v_i}^{\overline{p}}) \; \mathcal{L} \; r(I_{v_j}^{\overline{p}})$ for all
$1\leq i<j\leq |\Pi_{\mathcal{T}}|$.
Our dynamic programming algorithm uses
a binary table \[DT(i,j,s,s_1,s_2,s_3,s_4,s_5,s_6,s_{i,1},s_{i,2},s_{i,3}),\]
where $i,j,s,s_1, \dots ,s_6,s_{i,1},s_{i,2},s_{i,3}$ are non-negative integers and $DT(i,j,s,s_1, \dots ,s_6,s_{i,1},s_{i,2},s_{i,3})$ $=1$
if there is a submultiset $\Pi_{\mathcal{T}'}\sqsubseteq \{\pi_{v_1}, \pi_{v_2}, \dots, \pi_{v_i}\}$ such that
\begin{enumerate}
\item $|\Pi_{\mathcal{T}'}|=j$;
\item $\pi_{v_i}\in_+ \Pi_{\mathcal{T}'}$;
\item $\max \{SC_{\mathcal{V}\cup \mathcal{T}'}(c)\mid c\in \mathcal{C}\setminus \{p\}\}=s$;
\item $SC_{\mathcal{V}\cup \mathcal{T}'}(c_{t})=s_t$ for all $1\leq t\leq 6$, where $c_{3}=\overleftarrow{p}(1), c_{2}=\overleftarrow{p}(2), c_{1}=\overleftarrow{p}(3), c_{4}=\overrightarrow{p}(1), c_{5}=\overrightarrow{p}(2)$ and $c_{6}=\overrightarrow{p}(3)$; and
\item $SC_{\mathcal{V}\cup \mathcal{T}'}(c_{i,t})=s_{i,t}$ for all $t\in \{1,2,3\}$, where $c_{i,1}=r(I_{v_i}^{\overline{p}}), c_{i,2}=\overleftarrow{c_{i,1}}(1)$ and
$c_{i,3}=\overleftarrow{c_{i,1}}(2)$. See \myfig{fig:mapping} for an illustration of (4) and (5).
\end{enumerate}

\begin{figure}[h!]
\begin{center}
\includegraphics[width=\textwidth]{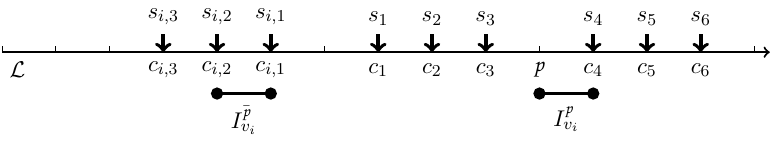}
\caption[Dynamic table for {\ccav}-{\appc}-Approval in 2-peaked elections]{Illustration of (4) and (5) in the definition of the dynamic table ${DT}$ in the proof of Theorem~\ref{3app}.}
\label{fig:mapping}
\end{center}
\end{figure}

It is easy to see that the subinstance is a {\yesins} if and only if \[DT(n,R',s,s_1,s_2,\dots,s_6,s_{n,1},s_{n,2},s_{n,3})=1\] for some $n\leq |\Pi_{\mathcal{T}}|$, $R'\leq R$, $s\leq SC_\mathcal{V}(p)+R'-1$ and $s'\leq s$ for all $s'\in \{s_1,s_2, \dots ,s_6,s_{n,1},s_{n,2},s_{n,3}\}$. Therefore, to solve the subinstance we need to calculate
the values of $DT(i,j,s,s_1,s_2,\dots,s_6,s_{i,1},s_{i,2},s_{i,3})$ for all $0\leq j\leq R$, $j\leq i\leq |\Pi_{\mathcal{T}}|$,
$1\leq s\leq SC_{\mathcal{V}}(p)+R-1$ and $s'\leq s$ for all $s'\in\{s_1, \dots ,s_6,s_{i,1},s_{i,2},s_{i,3}\}$.
As a result, we have at most $|\mathcal{T}|\cdot R\cdot(|\mathcal{V}|+R)^{10}$ entries to calculate.

We use the following recurrence relation to update the table: $DT(i, j, s, s_1, \dots, s_6, s_{i,1}, s_{i,2}, s_{i,3})=1$ if at least one of the following cases applies:

{Case 1.}
$\exists DT(i_1,j-1,s,s'_1,s'_2,\dots,s'_6,s'_{i_1,1},s'_{i_1,2},s'_{i_1,3})=1$ such that Conditions (1)-(4) hold.

{Case 2.} $\exists s'\in \{s_1, \dots ,s_6,s_{i,1}, \dots ,s_{i,3}\}$ with $s'=s$ and $\exists DT(i_1,j-1,s-1,s'_1, \dots ,s'_6,s'_{i_1,1}, \dots , s'_{i_1,3})=1$ such that Conditions (1)-(4) hold.

The four conditions are:
\begin{enumerate}
\item $j-1\leq i_1\leq i-1$;
\item $s_t=s'_t+SC_{\{v_i\}}(c_{t})$ for all $1\leq t\leq 6$;
\item $s_{i,t}=s'_{i_1,t_1}+SC_{\{v_i\}}(c_{i,t})$ for all $c_{i,t}=c_{i_1,t_1}$; and
\item $s_{i,t}=SC_{\mathcal{V}\cup \{v_i\}}(c_{i,t})$ for all $c_{i,t}\in \{{r(I^{\bar{p}}_{v_i})},\overleftarrow{r(I^{\bar{p}}_{v_i})}(1),$ $ \overleftarrow{r(I^{\bar{p}}_{v_i})}(2)\}\setminus \{{r(I^{\bar{p}}_{v_{i_{1}}})},\overleftarrow{r(I^{\bar{p}}_{v_{i_{1}}})}(1),$ $\overleftarrow{r(I^{\bar{p}}_{v_{i_{1}}})}(2)\}$.
\end{enumerate}

The above algorithm can be adapted to solve the nonunique-winner model: replacing all appearances of ``$SC_{\mathcal{V}}({\mathdiscandi})+\mathessize-1$" in the above description with ``$SC_{\mathcal{V}}({\mathdiscandi})+\mathessize$".

The algorithm can be easily generalized to solve {\appc}-AV-2 for every constant $\mathappc \geq 4$ by using a bigger but
still polynomial-sized table. In particular, for each fixed {\appc}, we need a $3{\mathappc}$-dimensional table \[DT(i, j, s, s_1, \dots , s_{2(\mathappc-1)}, s_{i,1}, \dots, s_{i,{\mathappc-1}}),\] where $i, j, s$ take the same meanings as in the above algorithm, $s_1, \dots ,s_{2(\mathappc-1)}$ maintain the scores of the $2(\mathappc-1)$ candidates around the distinguished candidate {\discandi} (precisely, we maintain the scores of the ${\mathappc}-1$ candidates immediately lying on the left side of {\discandi}, and the scores of the ${\mathappc}-1$ candidates immediately lying on the right side of {\discandi} in the 2-harmonious order. If there are less than $r-1$ candidates lying on the left- or right-side of {\discandi}, we reduce the dimension of the table accordingly), and $s_{i,1}, \dots ,s_{i,{\mathappc-1}}$ maintain the scores of the candidate $r(I_{v_i}^{\overline{\mathdiscandi}})$ and the $\mathappc-2$ candidates consecutively lying on the left side of $r(I_{v_i}^{\overline{\mathdiscandi}})$  in the 2-harmonious order.

\subsection{Proof of Theorem~\ref{rav2nphard}}
We prove Theorem~\ref{rav2nphard} by a reduction from a variant of the {\sc{Independent Set}} problem which is {\nph}~\cite{Keil99}. It is clear that $r$-AV-2 is in {\np}. It remains to prove the {\nphns}.

Let $(\;)$ denote an empty order containing no element.
For a linear order $A=(a_1, a_2, \dots, a_n)$, 
let $A[a_i, a_j]$ (resp. $A(a_i, a_j]$, $A[a_i, a_j)$ and $A(a_i, a_j)$) with $1\leq i\leq j\leq n$ be the sub-order $(a_i, a_{i+1}, \dots, a_j)$~(resp. $(a_{i+1}, a_{i+2}, \dots, a_j)$ if $i<j$ and $(\;)$ if $i=j$, $(a_i, a_{i+1}, \dots, a_{j-1})$ if $i<j$ and $(\;)$ if $i=j$, and $(a_{i+1}, a_{i+2}, \dots, a_{j-1})$ if $i<j-1$ and $(\;)$ if $j\geq i\geq j-1$), and let $A[a_j, a_i]$ (resp. $A[a_j, a_i)$, $A(a_j, a_i]$ and $A(a_j, a_i)$)
be the reversal of $A[a_i, a_j]$ (resp. $A(a_i, a_j]$, $A[a_i, a_j)$ and $A(a_i, a_j)$).
For two linear orders $A=(a_1, a_2, \dots, a_n)$ and $B=(b_1, b_2, \dots, b_m)$ without common elements, let $(A,B)$ be the linear order $(a_1, a_2, \dots, a_n, b_1, b_2, \dots, b_m)$. Let $[n]$ be the set $\{1, 2, \dots, n\}$.

\begin{quote}
\noindent{A Variant of Independent Set~(VIS)}\\[1mm]
\noindent{\it Input:} A multiset $\mathcal{T}=\{T_1,T_2, \dots ,T_n\}$ where each $T_i\in_+ \mathcal{T}$
is a set of discrete intervals of size 4 over $(1,2, \dots ,12n)$ and $|T_i|\leq 3$ for all $T_i\in_+ \mathcal{T}$.\\
\noindent{\it Question:} Is there a set $S\subseteq \bigcup_{T\in_+ \mathcal{T}}{T}$ of discrete intervals such that
$|S|=n$, $|S\cap T_i|=1$ for every $T_i\in_+ \mathcal{T}$ and no two discrete intervals in $S$ intersect?
\end{quote}

Given an instance $\mathcal{E}=(\mathcal{T}=\{T_1, T_2, \dots, T_n\})$ of VIS, we construct an instance $\mathcal{E}'=((\mathcal{C}, \Pi_{\mathcal{V}}), p\in \mathcal{C}, \Pi_{\mathcal{T}}, \mathcal{L}, R=n)$
for {\appc}-AV-2 as follows.

Let $\mathcal{I}=\bigcup_{T\in_+ \mathcal{T}}T$.
Let $\Gamma$ be the set of all elements appearing in some discrete interval of $\mathcal{I}$, i.e., $\Gamma=\bigcup_{I\in \mathcal{I}}\mathcal{S}({I})$.
Let $\vec{\Gamma}=(x_1, x_2, \dots, x_{|\Gamma|})$ be the order of $\Gamma$ such that $x_i< x_{i+1}$ for all $i\in [|\Gamma|-1]$.

{\bf{Candidates:}} We create three disjoint subsets of candidates $C$, $D$ and $E$ as follows:
(1) $C=\Gamma$; 
(2) $D$ contains exactly $2n-1$ candidates denoted by $d_1, \dots, d_n, \dots, d_{2n-1}$;
(3) $E$ contains exactly $(n+3)\cdot (|C|+|D|-1)$ dummy candidates denoted by $x_1', x_2', \dots , x_{|C|\cdot (n+3)}', d_1', d_2', \dots , d_{(n+3)\cdot (|D|-1)}'$ which can never be winners no matter which up to $R$ unregistered votes are added. Hence, $\mathcal{C}=C\cup D\cup E$. The distinguished candidate is $d_n$, i.e., $p=d_n$. Moreover, $\mathappc=n+4$.

{\bf{2-Harmonious Order:}} Let $\vec{D}=(d_1, d_2, \dots ,d_{2n-1})$ and $\vec{E}=(x_1', \dots ,x_{|C|\cdot (n+3)}', d_1', \dots ,d_{(|D|-1)\cdot (n+3)}')$. Then, the 2-harmonious order $\mathcal{L}$ is $(\vec{\Gamma},\vec{D}, \vec{E})$.

{\bf{Registered Votes $\Pi_{\mathcal{V}}$:}} We create the following registered votes:

(1) for each $x_i\in C$, create $n-2$ votes defined as \[(x_i,\mathcal{L}[x_{(n+3)i-n-2}',x_{i(n+3)}'],\mathcal{L}(x_i,x_1],\mathcal{L}(x_i,x_{(n+3)i-n-2}'), \mathcal{L}(x_{i(n+3)}',d_{(|D|-1)\cdot (n+3)}']);\]

(2) for each $d_i\in D$ where $i\in [n-1]$, create $n-(i+1)$ votes defined as \[(d_i,\mathcal{L}[d_{(n+3)i-n-2}',d_{i(n+3)}'],\mathcal{L}(d_i,x_1],\mathcal{L}(d_i,d_{(n+3)i-n-2}'), \mathcal{L}(d_{i(n+3)}',d_{(|D|-1)\cdot (n+3)}']);\]

(3) for each $d_i\in D$ where $i\in \{n+1, n+2, \dots, 2n-1\}$, create $i-(n+1)$ votes defined as \[(d_i,\mathcal{L}[d_{(n+3)i-2n-5}',d_{(i-1)\cdot (n+3)}'],\mathcal{L}(d_i,x_1], \mathcal{L}(d_i,d_{(n+3)i-2n-5}'), \mathcal{L}(d_{(i-1)\cdot (n+3)}',d_{(|D|-1)\cdot (n+3)}']).\]

{\bf{Unregistered Votes $\Pi_{\mathcal{T}}$:}}
For each $I_{ij}\in T_i\in_+ \mathcal{T}$, create a corresponding
unregistered vote which is defined as
$(\mathcal{L}[l(I_{ij}), r(I_{ij})], \mathcal{L}[d_i, d_{(|D|-1)\cdot (n+3)}'], \mathcal{L}{(l(I_{ij}), x_1]}, \mathcal{L}(r(I_{ij}), d_{i-1}])$.
Clearly, this vote approves exactly the four candidates between $l(I_{ij})$ and $r(I_{ij})$ (including $l(I_{ij})$ and $r(I_{ij})$), and all candidates  between $d_i$ and $d_{i+n-1}$ (including $d_i$ and $d_{i+n-1}$) in $\mathcal{L}$.
Thus, every unregistered vote approves the distinguished candidate $d_n$.

It is clear that all votes constructed above are 2-peaked with respect to $\mathcal{L}$.
Due to the construction, it is easy to see that $SC_{\mathcal{V}}(c)=n-2$ for all $c\in C$,
$SC_{\mathcal{V}}(d_i)=n-i-1$ for all $d_i\in D$ where $i\in [n-1]$, $SC_{\mathcal{V}}(d_i)=i-n-1$ for
all $d_i\in D$ where $i\in \{n+1, n+2, \dots, 2n-1\}$, $SC_{\mathcal{V}}(c)\leq n-2$ for all $c\in E$, and $SC_{\mathcal{V}}(d_n)=0$.

Now we prove that $\mathcal{E}$ is a {\yesins} if and only if $\mathcal{E}'$ is a {\yesins}.

($\Rightarrow:$) Suppose that $\mathcal{E}$ is a {\yesins} and let $S$ be a solution of $\mathcal{E}$.
Let $\vec{S}=(I_1, I_2, \dots, I_n)$ be the order of $S$ where $I_i=S\cap T_i$ for all $i\in [n]$.
Then, we can make $d_n$ the unique winner by adding votes from $\Pi_{\mathcal{T}}$ according to $S$.
More specifically, for each $I_i\in S$ we add its corresponding vote constructed above to the election.
Clearly, the final score of $d_n$ is $n$.
Due to the construction, no two added votes $\pi_v$ and $\pi_u$ which correspond to two different intervals $I_i$ and $I_j$, respectively, approve a common candidate from $C$. Thus, after adding these votes, no candidate in $C$ has a higher score than that of $d_n$.
To analyze the score of $d_j\in D$ with $j\in[n-1]$, observe that for any $i>j$
the vote corresponding to $I_{i}$ does not approve $d_j$.
Since $SC_{\mathcal{V}}(d_j)=n-j-1$ and $|S\cap T_{i}|=1$ for all $i\in [j]$, the final score of $d_j$ is less than $n$.
Similarly, to analyze the score of $d_j\in D$ with $j\in \{n+1, n+2, \dots, 2n-1\}$, observe that for any $i\leq j-n$
the vote corresponding to $I_{i}$ does not approve $d_j$.
Since $SC_{\mathcal{V}}(d_j)=j-n-1$ and $|S\cap T_i|=1$ for all $i\in \{j-n+1,j-n+2, \dots ,n\}$,
the final score of $d_j$ is less than $n$. The final score of each $c\in E$ is clearly at most $n-2$ since no unregistered vote approves $c$.
To summarize the above analysis, we conclude
that the distinguished candidate $d_n$ becomes the unique winner after adding the votes discussed above.

($\Leftarrow:$) Suppose that $\mathcal{E}'$ is a {\yesins} and $S'$ is a multiset of votes
chosen from $\Pi_{\mathcal{T}}$ such that $d_n$ uniquely wins $(\mathcal{C}, \Pi_{\mathcal{V}}\uplus S')$.
It is easy to verify that $|S'|=n$, since otherwise at least one of $C$ would be a winner. Thus, the final score of $d_n$ is $n$ and every $c\in C$ can get at most one point from $S'$. Therefore, no two votes in $S'$ approve a common candidate of $C$, implying that $S'$ must be a set. This also means that the intervals corresponding to $S'$ do not intersect.
Let $P_1, P_2, \dots, P_n$ be a partition of $\Pi_{\mathcal{T}}$
where $P_i$ contains all votes corresponding to the intervals of $T_i\in_+ \mathcal{T}$. Clearly, $P_i$ is a
set. We claim that $|S'\cap P_i|=1$ for every $i\in [n]$. Suppose this is not true, then there must be a certain
$P_i$ with $|S'\cap P_i| \geq 2$. Let $S_1=S'\cap P_i$~(thus, $|S_1|\geq 2$), $S_2=\{\pi_v\in S'\cap P_{i'}\mid i'<i\}$ and
$S_3=\{\pi_v\in S'\cap P_{i'}\mid i'>i\}$. It is clear that $|S_1|+|S_2|+|S_3|=n$.
Since all votes in $S_1$ approve both $d_i$ and $d_{i+n-1}$, all votes in $S_2$ approve $d_i$ but do not approve $d_{i+n-1}$, and all votes in $S_3$ approve $d_{i+n-1}$
but do not approve $d_i$, it holds that

\[
\begin{array}{l}
 SC_{\mathcal{V}\uplus S'}(d_i)+SC_{\mathcal{V}\uplus S'}(d_{i+n-1})\\
 =SC_{\mathcal{V}}(d_i)+|S_1|+|S_2|+SC_{\mathcal{V}}(d_{i+n-1})+|S_1|+|S_3|\\
 = n-i-1+|S_1|+|S_2|+i-2+|S_1|+|S_3|\\
 = 2n-3+|S_1|\\
 \geq 2n-1
\end{array}
\]

As a result, at least one of $d_i$ and $d_{i+n-1}$ has final score at least $n$, contradicting that $d_n$ is the unique winner.
The claim then follows.
It is now easy to see that the discrete intervals corresponding to $S'$ form a solution of $\mathcal{E}$.

The {\nphns} reduction for the nonunique-winner model is similar to the above reduction, with only the difference in the construction of the registered votes. In particular, we need to construct the registered votes so that the score of each candidate $c\in C\cup D$ is exactly one point greater than that of $c$ in the above construction. This can be done as follows:

(1) for each $x_i\in C$, create $n-1$ votes defined as
\[(x_i,\mathcal{L}[x_{(n+3)i-n-2}',x_{i(n+3)}'],\mathcal{L}(x_i,x_1], \mathcal{L}(x_i,x_{(n+3)i-n-2}'),\mathcal{L}(x_{i(n+3)}',d_{(|D|-1)\cdot (n+3)}']);\]

(2) for each $d_i\in D$ where $i\in [n-1]$, create $n-i$ votes defined as
\[(d_i,\mathcal{L}[d_{(n+3)i-n-2}',d_{i(n+3)}'],\mathcal{L}(d_i,x_1],\mathcal{L}(d_i,d_{(n+3)i-n-2}'),\mathcal{L}(d_{i(n+3)}',d_{(|D|-1)\cdot (n+3)}']);\]

(3) for each $d_i\in D$ where $i\in \{n+1, n+2, \dots, 2n-1\}$, create $i-n$ votes defined as
\[(d_i,\mathcal{L}[d_{(n+3)i-2n-5}',d_{(i-1)\cdot (n+3)}'],\mathcal{L}(d_i,x_1], \mathcal{L}(d_i,d_{(n+3)i-2n-5}'), \mathcal{L}(d_{(i-1)\cdot (n+3)}',d_{(|D|-1)\cdot (n+3)}']).\]

\subsection{Proof of Theorem~\ref{theorem:3ad2}}
\label{subsection:1}
It is clear that $r$-DV-2 is in {\np}. It remains to prove the {\nphns}. We first prove that 3-DV-2 is {\nph} by a reduction from the {\sc{Vertex Cover}} problem on graphs of maximum degree 3 which is {\nph}~\cite{DBLP:journals/siamam/GareyJ77}.
Then, we will show that the proof applies to {\appc}-DV-2 for every constant $r\geq 4$ with a slight modification.

A {\it{graph}} is a tuple $G=(V, E)$ where $V$ is the set of {\it{vertices}} and $E$
is the set of {\it{edges}}. We also use $V(G)$ to denote the vertex set of $G$. For a vertex $u\in V$, $N_G(u)$ denotes the set of its {\it{neighbors}} in $G$, i.e., $N_G(u)=\{w\mid (w,u)\in E\}$. The {\it{degree}} of a vertex is the number of its neighbors. A vertex of degree $i$ is called a {\it{degree-$i$}} vertex.
A graph is of {\it{maximum degree $3$}} if every vertex has degree at most $3$. An {\it independent set} of a graph $G=(V,E)$
is a subset $S\subseteq V$ such that there is no edge between every two vertices in $S$. Meanwhile, the complement $V\setminus S$ is called a {\it{vertex cover}} of $G$.

\begin{quote}
\noindent{Vertex Cover on Graphs of Maximum Degree 3 (VC3)}\\
\noindent\textit{Input:} A graph $G=(V, E)$ of maximum degree 3 and a positive integer $\kappa$.\\
\noindent\textit{Question:} Does $G$ have a vertex cover of size at most $\kappa$?
\end{quote}

To prove the {\nphns} of 3-DV-2, we first study a 
property of graphs of maximum degree 3. This property may be of independent interest since many graph problems are {\nph} when restricted to graphs of maximum degree 3 (see e.g.,~\cite{garey,DBLP:journals/tcs/GareyJS76,DBLP:journals/tcs/Picouleau94}).

An {\textit{interval}} over the real line is a closed set $[a, b]=\{x\in \mathbb{R}\mid a\leq x\leq b\}$ where $a$ and $b$ are real numbers. 
An interval is {\it trivial} if $a=b$. For an interval $I=[a,b]$, $l(I)$ denotes its left-endpoint $a$, and $r(I)$ denotes its right-endpoint $b$.
A {\it $t$-interval} is a set of $t$ intervals over the real line. A {\it{$t_{\leq}$-interval}} is a $t'$-interval for some $t'\leq t$.
A graph $G=(V, E)$ is a {\it t-interval graph} if there is a set
$\mathcal{I}(G)$ of $t_{\leq}$-intervals and a bijection $f: V\rightarrow \mathcal{I}(G)$ such that for every $u, w \in V$, $(u, w)\in E$ if and only if $f(u)$ and $f(w)$ intersect.
Here, $\mathcal{I}(G)$ is called a {\it t-interval representation}
of $G$. 
For simplicity, we use $\mathcal{I}_u(G)=\{I_u^1, I_u^2, \dots ,I_u^{t'}\}$ to denote $f(u)$, where each $I_u^i\in \mathcal{I}_u(G)$ is an interval and $t'\leq t$. Moreover, when it is clear from the context, we write $\mathcal{I}_u$ for $\mathcal{I}_u(G)$. For two real numbers $a$ and $b$ with $a\leq b$, we define $(a,b)=\{x\in \mathbb{R}\mid a<x<b\}$.

The following lemma states that every graph of maximum degree 3 has a 2-interval representation such that every vertex is represented by a 2-interval where one interval is trivial and, moreover, 2-intervals may only intersect at the endpoints.

\begin{lemma}
\label{lemma:ids}
For every graph $G$ of maximum degree 3 there is a 2-interval representation such that for every $u\in V(G)$,
$\mathcal{I}_u=\{I_u^1, I_u^2\}$ and one of the following conditions holds:
\begin{enumerate}
\item $I_u^1=[x_1, x_1], I_u^2=[x_2, x_3], x_1<x_2<x_3$ and
$\nexists u'\in V(G)\setminus \{u\}$ such that $r(I(u'))\in (x_2, x_3)$~or~$l(I(u'))\in (x_2, x_3);$
\item $I_u^1=[x_1,x_2], I_u^2=[x_3,x_3], x_1<x_2<x_3$ and
$\nexists u'\in V(G)\setminus \{u\}$ such that $r(I(u'))\in (x_1, x_2)$ or $l(I(u'))\in (x_1, x_2),$ for some $I(u')\in\{I^1_{u'}, I^2_{u'}\}.$
\end{enumerate}
Moreover, such a $2$-interval representation can be found in polynomial time. See \myfig{figgraphproperty} for an example.
\end{lemma}

\begin{figure}[h!]
\begin{center}
\includegraphics[width=\textwidth]{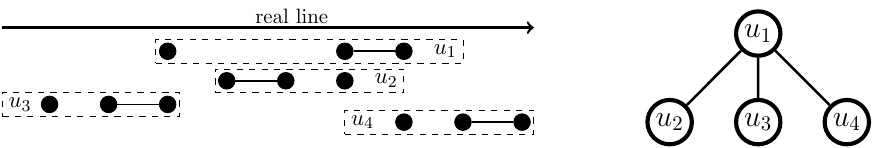}
\caption[A property of graphs with maximum degree 3]{The figure on the left-side illustrates a 2-interval representation of the graph on the right-side.}
\label{figgraphproperty}
\end{center}
\end{figure}

The proof of Lemma~\ref{lemma:ids} is deferred to the appendix. We now show the reduction.
Let $\mathcal{E}=(G, \kappa)$ be an instance of VC3 and $\mathcal{I}(G)$ be a 2-interval representation of $G$
as stated in Lemma~\ref{lemma:ids}. For every $\mathcal{I}_u=\{I^1_u, I^2_u\}$,
let $D(u)$ be the set of the endpoints of $I^1_u$ and $I^2_u$~(due to Lemma~\ref{lemma:ids}, $|D(u)|=3$ for all $u\in V(G)$), and let
$\Gamma=\bigcup_{u\in V(G)}{D(u)}$. Let $\vec{\Gamma}=(x_1, x_2, \dots, x_{|\Gamma|})$ be
the order of $\Gamma$ with $x_i< x_{i+1}$ for all $i\in [|\Gamma|-1]$.
We construct an instance
$\mathcal{E}'=((\mathcal{C}, \Pi_{\mathcal{V}}), p\in \mathcal{C}, \mathcal{L}, R=\kappa)$ of 3-DV-2 as follows.

{\bf{Candidates:}} $\mathcal{C}=\Gamma \cup \{p, c_1, c_2, c_3, c_4\}$ with $c_1, c_2, c_3, c_4$ being dummy candidates, which would never be winners no matter which up to $R$ votes are deleted. Here, $\{p, c_1, \dots ,c_4\}$ is disjoint with $\Gamma$.

{\bf{2-Harmonious Order:}} $\mathcal{L}=(\vec{\Gamma}, p, c_1, c_2, c_3, c_4)$.

{\bf{Votes:}} We create in total $|V(G)|+2$ votes. First, we create $|V(G)|$ votes each of which corresponds to an $\mathcal{I}_u$ in $\mathcal{I}(G)$ for $u\in V(G)$. In particular, for every $\mathcal{I}_u$ where $u\in V(G)$, we create a vote \[\pi_u=(x_i,x_j,x_k,\mathcal{L}(x_{i},x_1],\mathcal{L}(x_i,x_j),\mathcal{L}(x_j,x_k),\mathcal{L}(x_k,c_4]),\] where $(x_i, x_j, x_k)$ is the order of $D(u)$ such that $x_i<x_j<x_k$. Apparently, $\pi_u$ approves all candidates in $D(u)$ and disapproves $p$.
Due to Lemma~\ref{lemma:ids}, either $x_i$ or $x_k$ lies consecutively with $x_j$ in $\mathcal{L}$, which implies that $\pi_u$ is 2-peaked with respect to $\mathcal{L}$. Let $\Pi(G)$ be the multiset of the above $|V(G)|$ votes.
Second, we create two further votes defined as
$(p,c_1,c_2,c_3,c_4,\mathcal{L}(p, x_1])$ and
$(p,c_3,c_4,c_1,c_2,\mathcal{L}(p, x_1])$, respectively.
It is clear that these two votes are 2-peaked with respect to $\mathcal{L}$.

In the following, we prove that $\mathcal{E}$ is a {\yesins}
if and only if $\mathcal{E}'$ is a {\yesins}.

$(\Rightarrow:)$ Suppose that $\mathcal{E}$ is a {\yesins} and $S$ is a vertex cover of
size at most $\kappa$ of $G$. Then, we delete all
votes in $\{\pi_u\mid u\in S\}$. After deleting these votes, no
two votes of $\Pi(G)$ approve a common candidate in $\mathcal{C}$, since otherwise $V(G)\setminus S$ could not be an independent set, contradicting the fact that
$S$ is a vertex cover. Thus, after deleting these votes all candidates
except for $p$ have only one point. Since $p$ has two points, $p$ is the unique winner.

$(\Leftarrow:)$ Suppose that $\mathcal{E}'$ is a {\yesins}. Then $\mathcal{E}'$ has a solution containing only votes disapproving $p$. Let $S'$ be such a solution of size at most $\kappa$.
Therefore, $p$ has two points in the election after deleting all votes in $S'$. It follows that every other candidate has at most one
point after deleting all votes of $S'$. As a result, no two votes of $\Pi(G)$ approve
a common candidate in $\mathcal{C}$ in the final election, implying that the vertices corresponding to
$S'$ form a vertex cover of $G$.

In order to prove that {\appc}-DV-2 is {\nph} for a constant $r\geq 4$, we need to modify the proof slightly. First, we add some dummy candidates. More specifically, for every $u\in V(G)$ such that $[x_i,x_{i+1}]\in \mathcal{I}_u$ where $1\leq i\leq |\Gamma|-1$, we create ${\mathappc}-3$ dummy candidates and put them between $x_i\in \Gamma$ and $x_{i+1}\in \Gamma$ in $\mathcal{L}$ (the relative order of these dummy candidates in $\mathcal{L}$ does not matter).
Besides, we have $2\mathappc-6$ dummy candidates $c_5,c_6, \dots ,c_{2\mathappc-2}$ lying after $c_4$ in $\mathcal{L}$, with the order $(c_5,c_6, \dots ,c_{2\mathappc-2})$.
Thus, there are ${({\mathappc}-3)}\cdot |V(G)|+2\mathappc-6$ new dummy candidates in total.
We change the votes in $\Pi(G)$ as follows: for every $u\in V(G)$ with $\mathcal{I}_u=\{[x_i,x_{i+1}],[x_j,x_j]\}, i+1<j$ (resp. $\mathcal{I}_u=\{[x_i,x_i],[x_j,x_{j+1}]\}, i<j$), we create a vote defined as $(\mathcal{L}[x_i,x_{i+1}],x_j,\mathcal{L}(x_{i},x_1],\mathcal{L}(x_{i+1},x_j), \mathcal{L}(x_j,c_{2\mathappc-2}])$ (resp. $(x_i,\mathcal{L}[x_j,x_{j+1}],\mathcal{L}(x_{i},x_1],\mathcal{L}(x_{i},x_j), \mathcal{L}(x_{j+1},c_{2\mathappc-2}])$).
As for the last two created votes, we replace them with the two votes defined as $(\mathcal{L}[\mathdiscandi,c_{2{\mathappc}-2}],\mathcal{L}(\mathdiscandi,x_1])$ and $(\mathdiscandi,\mathcal{L}[c_{\mathappc},c_{2\mathappc-2}], \mathcal{L}[c_1,c_{\mathappc}),\mathcal{L}(\mathdiscandi,x_1])$, respectively. Then, with the same argument, we can show that {\appc}-DV-2 is {\nph}.

To prove the {\nphns} of the nonunique-winner model of {\appc}-DV-2 for every constant $\mathappc \geq 3$, we adapt the above reductions slightly. In particular, we keep all votes in $\Pi(G)$ and the second-last vote created above, but discard the last vote created above (so that the score of the distinguished candidate {\discandi} is 1 in the given election). Other parts remain unchanged.

\section{3-Peaked Elections}
In Section 2, we proved that for every constant $r$ control by adding votes for {\appc}-Approval is polynomial-time solvable when restricted to 2-peaked elections.
In this section, we show that the tractability of the problem does not hold when extended to 3-peaked elections, even when $r$ is a small constant.

\begin{theorem}
\label{theorem:4acav3p}
{\appc}-AV-3 is {\npc} for every constant $r\geq 4$.
\end{theorem}
\begin{proof}
It is clear that {\appc}-AV-3 is in {\np}. It remains to prove the {\nphns} of the problem. We first prove the {\nphns} of 4-AV-3 by a reduction from the {\sc{Independent Set}} problem on graphs of maximum degree 3 which is {\nph}~\cite{DBLP:journals/siamam/GareyJ77}.

\begin{quote}
\noindent Independent Set on Graphs of Maximum Degree 3 (IS3)\\
\noindent{\it Input:} A graph $G=(V,E)$ of maximum degree 3 and a positive integer $\kappa$.\\
\noindent{\it Question:} Does $G$ have an independent set containing exactly $\kappa$ vertices?
\end{quote}

For an instance $\mathcal{E}=(G, \kappa)$ of IS3,
let $\mathcal{I}(G)$ be a 2-interval representation of $G$ which satisfies all conditions stated in Lemma~\ref{lemma:ids}.
Let $D(u), \Gamma$ and $\vec{\Gamma}$ be defined as in Subsect.~\ref{subsection:1}.
We construct an instance $\mathcal{E}'=((\mathcal{C}, \Pi_{\mathcal{V}}), p \in \mathcal{C}, \Pi_{\mathcal{T}}, \mathcal{L}, R=\kappa)$ of 4-AV-3 as follows.

{\bf{Candidates:}} $\mathcal{C}=\Gamma \cup \{p, c_1, c_2, c_3\}$, where $\{p, c_1, c_2, c_3\}$ is disjoint from $\Gamma$.

{\bf{3-Harmonious Order:}} $\mathcal{L}=(\vec{\Gamma}, p, c_1, c_2, c_3)$.

{\bf{Registered Votes $\Pi_{\mathcal{V}}$:}} The role of registered votes is to ensure that all candidates of $\Gamma$ have the same score $\kappa-2$.
To this end, we first create $\kappa-2$ votes defined as
$(\mathcal{L}[x_i, x_{i+3}], \mathcal{L}(x_i, x_1],\mathcal{L}(x_{i+3}, c_3])$
for every $i=1,5,\dots, 4\lfloor |\Gamma|/4\rfloor-3$. Then, we create
some further votes according to $|\Gamma|$.
\noindent\textbf{Case 1.} $|\Gamma|\equiv 0 \mod 4$. We create no
further vote.
\noindent\textbf{Case 2.} $|\Gamma|\equiv 1 \mod 4$. We create
additional $\kappa-2$ votes defined as
$(x_{|\Gamma|}, \mathcal{L}[c_1, c_3], \mathcal{L}(x_{|\Gamma|}, x_1], p)$.
\noindent\textbf{Case 3.} $|\Gamma|\equiv 2 \mod 4$. We create
additional $\kappa-2$ votes defined as
$(x_{|\Gamma|-1}, x_{|\Gamma|}, c_1, c_2, \mathcal{L}(x_{|\Gamma|-1}, x_1], p, c_3)$.
\noindent\textbf{Case 4.} $|\Gamma|\equiv 3 \mod 4$. We create
additional $\kappa-2$ votes defined as $(\mathcal{L}[x_{|\Gamma|-2}, x_{|\Gamma|}], c_1, \mathcal{L}(x_{|\Gamma|-2}, x_1], p, c_2, c_3)$.

{\bf{Unregistered Votes $\Pi_{\mathcal{T}}$:}} 
For each $u\in V(G)$, let $(x_i, x_j, x_k)$ be the order of $D(u)$ with $x_i<x_j<x_k$. We create a
vote $\pi_u=(x_i, x_j, x_k, p, \mathcal{L}(x_{i},x_1], \mathcal{L}(x_i, x_{j}), \mathcal{L}(x_{j},x_k), \mathcal{L}(x_k,p), \mathcal{L}(p,c_3])$.
Due to Lemma~\ref{lemma:ids}, either $x_i$ or $x_k$ lies consecutively with $x_j$ in $\mathcal{L}$; thus, in the worst case
$\pi_u$ has 3 peaks $x_{\alpha}, x_{\beta}$ and $p$ where $\{x_{\alpha}, x_{\beta}\}\subseteq \{x_i, x_j, x_k\}$
($\{x_{\alpha}, x_{\beta}\}$ depends on whether $x_j$ lies consecutively with $x_i$ or with $x_k$), with respect to $\mathcal{L}$.

In the following, we prove that $\mathcal{E}$ is a {\yesins} if and
only if $\mathcal{E}'$ is a {\yesins}.
It is easy to see that $SC_{\mathcal{V}}(x)=\kappa-2$ for all $x\in \mathcal{C}\setminus \{p, c_1, c_2, c_3\}$, $SC_{\mathcal{V}}(p)=0$
and $SC_{\mathcal{V}}(c)\leq \kappa-2$ for all $c\in  \{c_1, c_2, c_3\}$.

$(\Rightarrow:)$ Suppose that $\mathcal{E}$ is a {\yesins} and $S$ is an independent set of size $\kappa$.
Then we add all votes corresponding to $S$, i.e., all votes in $\{\pi_u\mid u\in S\}$, to the election.
\mbox{Since} $S$ is an independent set, no two added votes approve a common candidate except $p$;
thus, each candidate except $p$ has a final score at most $\kappa-1$.
Since each added vote approves $p$, it follows that $p$ has a final score $\kappa$,
implying that $p$ becomes the unique winner after adding these votes.

$(\Leftarrow:)$ Suppose that $\mathcal{E}'$ is a {\yesins} and $S'$ is a solution. Clearly, $p$ has a final score $\kappa$. Since $p$ is the unique winner, for every $c\in \mathcal{C}\setminus \{p\}$,
there is at most one vote in $S'$ approving $c$. Thus, no two votes in $S'$ approve a common candidate except $p$. Due to the construction, the vertices corresponding to $S'$ form an independent set.

The proof applies to {\appc}-AV-3 for every constant $r\geq 5$ by a similar modification as discussed in Subsect.~\ref{subsection:1}.
To prove the {\nphns} of the nonunique-winner model of {\appc}-AV-3 for every $\mathappc \geq 4$, we adapt the above reductions so that every candidate in $\Gamma$ has the same score $\mathindependentsetdthree-1$ other than $\mathindependentsetdthree-2$. This can be done by replacing all appearances of ``$\mathindependentsetdthree-2$'' with ``$\mathindependentsetdthree-1$'' in the above reductions. The correctness argument is similar.
\end{proof}

\section{Parameterized Complexity}
In this section, we study {\appc}-Approval control problems from the viewpoint of parameterized complexity.
The first two {\fpt} results are for the general case, i.e., each registered and unregistered vote (if applicable) is not necessarily {\kpeak}-peaked but can be defined as any linear order over the candidates.

\begin{theorem}
\label{theorem:3ad2fpt}
For every constant {\appc}, {\appc}-DV is {\fpt} with respect to the number of deleted votes.
\end{theorem}
\begin{proof}
We derive an {\fpt}-algorithm for {\appc}-DV to prove the theorem. Let $\mathcal{E}=((\mathcal{C},\Pi_{\mathcal{V}}),p\in \mathcal{C}, R)$ be a given instance. Let $\mathcal{C}_1=\{c\in \mathcal{C}\setminus \{p\}\mid SC_{\mathcal{V}}(c)\geq SC_{\mathcal{V}}(p)\}$ and $\Pi_{\mathcal{V}_1}=\{\pi_v\in_+ \Pi_{\mathcal{V}} \mid \mathcal{C}_1\cap 1(v)\neq \emptyset\}$. That is, $\Pi_{\mathcal{V}_1}$ contains all votes which approve at least one candidate having at least the same score as that of $p$. Observe that every {\yesins} of {\appc}-DV has a solution $S$ with $S\sqsubseteq \Pi_{\mathcal{V}_1}$.
Hence, we can restrict our attention to $\Pi_{\mathcal{V}_1}$.
Since at most $R$ votes can be deleted and each vote approves at most {\appc} candidates in $\mathcal{C}_1$, if $\mathcal{E}$ is a {\yesins}, it must be that $|\mathcal{C}_1|\leq r\cdot R$. So we assume that $|\mathcal{C}_1|\leq r\cdot R$; since otherwise we can immediately conclude that ${\mathcal{E}}$ is a {\noins}.
We say that two votes have the same {\it type} if they approve the same candidates in $\mathcal{C}_1$.
Clearly, there are at most $O((\mathappc\cdot R)^{\mathappc})=O(R^{\mathappc})$ different types of votes in $\Pi_{\mathcal{V}_1}$.
Since every solution includes at most $R$ votes from each type of votes, we have at most ${O(f(R))+R-1 \choose R}=O(f(R)^R)$ cases to check where $f(R)=R^{\mathappc}$, implying an {\fpt}-algorithm for {\appc}-DV with respect to $R$.
\end{proof}

\begin{theorem}\label{fptacav}
For every constant {\appc}, {\appc}-AV is {\fpt} with respect to the number of added votes.
\end{theorem}
\begin{proof}
To prove the theorem, we reduce {\appc}-AV to a generalized {\appc}-{\sc{Set Packing}} problem, and derive an {\fpt}-algorithm for this generalized {\appc}-{\sc{Set Packing}} problem.

In the {\it {\appc}-{\sc{Set Packing}}} problem, we are given a set $X$ of elements, a multiset $U$ of {\appc}-subsets of
$X$ and an integer $R\geq 0$, and the question is whether there is a subcollection $T$ of $U$ such that $|T|=R$ and every
element in $X$ occurs in at most one {\appc}-subset in $T$. Jia, Zhang and Chen~\cite{DBLP:journals/jal/JiaZC04setpacking}
derived an ingenious {\fpt}-algorithm for the {\appc}-{\sc{Set Packing}} problem for every constant {\appc}, with respect to $R$. The {\fpt}-algorithm in~\cite{DBLP:journals/jal/JiaZC04setpacking} resorts to the  greedy localization technique, which has been proved useful in solving many parameterized problems~\cite{DBLP:journals/talg/ChenLLSZ12,DBLP:conf/iwpec/DehneFRS04,FH+04,DBLP:journals/algorithmica/FellowsKNRRSTW08,DBLP:conf/iwpec/LiuLCS06}.
At a general level, the {\fpt}-algorithm in~\cite{DBLP:journals/jal/JiaZC04setpacking} first greedily computes a maximal collection of pairwise disjoint sets. If it is of size at least $R$, the instance is solved; otherwise, the maximal collection must involve at most $\mathappc\cdot R$ elements. Based on the observation that every set in a solution intersects some set in the maximal collection, the algorithm enumerates all ``potential'' solutions iteratively, where each iteration is based on a partial solution obtained via a ``local extension'' move from the previous iteration. More importantly, the number of enumerated ``potential'' solutions in each iteration is bounded by a function of the parameter (see~\cite{DBLP:journals/jal/JiaZC04setpacking} for further details).

Compared with the {\appc}-{\sc{Set Packing}} problem, the generalized
{\appc}-{\sc{Set Packing}} problem allows every element in $X$ to occur in several {\appc}-subsets in the solution. Let $\mathds{N}^+$ be the set of all positive integers.
The formal definition of the problem is as follows.

\begin{quote}
\noindent Multi-{\appc}-Set Packing (M{\appc}SP)\\[1mm]
\noindent\textit{Input:} A universal set $X$, a multiset $U$ of {\appc}-subsets of $X$ where each $x\in X$ occurs in at least one {\appc}-subset, a mapping $f: X\rightarrow \mathds{N}^+$ with $f(x)\leq x_U$ for every $x\in X$, where $x_U$ is the number of occurrences of $x$ in the {\appc}-subsets of $U$, and a positive integer $R$.\\[2mm]
\noindent\textit{Question:} Is there a submultiset $T$ of $U$ such that $|T|=R$ and each $x\in X$ occurs in at most $f(x)$ many {\appc}-subsets of $T$?
\end{quote}

We adopt the same notation system as in~\cite{DBLP:journals/jal/JiaZC04setpacking}.
We call a submultiset $T$ of $U$ an {\it{$f$-{\appc}-set packing}} if every $x\in X$ occurs in at most $f(x)$ many {\appc}-subsets in $T$.
An {\it $f$-{\appc}-set packing} $T$ is {\it maximal} if after adding any set from $U\ominus T$ to $T$,
some $x\in X$ occurs in more than $f(x)$ sets in $T$.
A {\it partial set} $\sigma^*$ is a set in $U$ with zero or more elements in the set replaced by the
symbol ``$*$''. A set without ``$*$'' is also called a {\it regular set}.
The set consisting of the non-$*$ elements in a partial set $\sigma^*$
is denoted as $reg(\sigma^*)$. A partial set $\sigma^*$ is {\it consistent with} a regular set $\sigma$ if $reg(\sigma^*)\subseteq \sigma$.
For an $f$-{\appc}-set packing $T$, let $\mathbb{S}(T)$ be the set of objects contained in some {\appc}-subset in $T$. For instance, for $T=\{\{1,2,3\},\{1,2,3\},\{2,4,6\},\{6,2,1\}\}$, $\mathbb{S}(T)=\{1,2,3,4,6\}$.
We say that an order $(x_{a(1)}, x_{a(2)}, \dots, x_{a(t)})$, where $x_{a(i)}\in X$ for every $i\in \{1,2, \dots ,t\}$ (it may be that $x_{{a(i)}}=x_{a(j)}$ for $i\neq j$), is {\it{valid}} if every $x\in X$ occurs at most $f(x)$ times in the order.
For example, for $X=\{1, 4, 6\}$, $f(1)=2, f(4)=1$ and $f(6)=4$, $(1, 4, 1)$ is a valid order but $(1, 4, 1, 6, 1)$ is not a valid order.
A multiset $T'$ of partial sets is {\it{valid}}
if every $x\in X$ occurs in at most $f(x)$ partial sets in $T'$.

Our algorithm is a natural generalization of the one for the {\appc}-{\sc{Set Packing}} problem studied in~\cite{DBLP:journals/jal/JiaZC04setpacking}. The main idea of our algorithm is as follows. We first find an arbitrary maximal $f$-{\appc}-set packing $T_0$. This can be done in polynomial time. If $|T_0|\geq R$, we are done. Otherwise, $|\mathbb{S}(T_0)|\leq \mathappc\cdot R$. Then, we try to extend all possible valid multisets $T_R$ of the form $\{\sigma_1^*=\{x_1,*, \dots ,*\},\sigma_2^*=\{x_2,*, \dots ,*\}, \dots ,\sigma_R^*=\{x_R,*, \dots ,*\}\}$ to a solution, where each $x_i, 1\leq i\leq R$, belongs to $\mathbb{S}(T_0)$ (it may be that $x_i=x_j$ for $i\neq j$). We make sure that if the given instance is a {\yesins}, then at least one extension leads to a solution. When extending a $T_R$, we first make a copy $Q_R$ of $T_R$, which is used to find a subset of $X$ such that, in the case that $T_R$ can be extended to a solution, at least one element in the subset is included in a solution extended from $T_R$. This is done as follows: for each regular set in $U$, we check if there is a consistent partial set in $Q_R$ such that replacing the partial set with the regular set does not lead $Q_R$ to be invalid. If this is the case, we replace the partial set in $Q_R$ with the regular set. If all partial sets of $Q_R$ can be replaced by this way, we are done. Otherwise, there is still a partial set $\sigma^*$ after the replacement. An observation is that if the current $T_R$ can be extended to a solution, then in every solution the partial set $\sigma^*$ is extended to a regular set where one of its elements is from $\mathbb{S}(Q_R)$. Due to this observation, we extend $T_R$ by enumerating all possibilities of the element in $\mathbb{S}(Q_R)$ that is put in $\sigma^*$ in the solution. As the size of $\mathbb{S}(Q_R)$ is bounded by $r\cdot R$, we have at most $r\cdot R$ possibilities to consider. After each extension of $T_R$, we reset $Q_R=T_R$ and continue to extend $T_R$ analogously. The procedure terminates until we arrive at a solution or all possibilities are enumerated.
A formal description of the algorithm is shown in Algorithm~\ref{alga}.

\begin{algorithm}
{
find an arbitrary maximal $f$-{\appc}-set packing $T_0$ in polynomial time\;
\If{$|T_0|\geq R$}{return {\yes}\;}
\If{there is no valid order $(x_{1}, x_{2},\dots,x_{R})$ over $\mathbb{S}(T_0)$}{return {\no}\;\tcc*[f]{The correctness of this step is based on the observation that each set in a solution of a {\yesins} intersects $\mathbb{S}(T_0)$: if this is not the case, then $T_0$ cannot be a maximal $f$-{\appc}-set packing, since we can add the set in the solution which does not intersect $\mathbb{S}(T_0)$ to $T_0$.}}
\ForEach{possible valid order $(x_{1}, x_{2}, \dots, x_{R})$ over $\mathbb{S}(T_0)$}
{
let $T_R=\{\sigma_1^*, \sigma_2^*, \dots, \sigma_R^*\}$, where each partial set $\sigma_i^*, 1\leq i\leq R$,
consists of the element $x_i$ and two ``$*$''s\;
return Extend($T_R$)\;
}
return {\no}\;
\caption{An {\fpt}-algorithm for M{\appc}SP}
\label{alga}}
\end{algorithm}

\begin{procedure}[h]
$Q_R=T_R$\;
\ForEach{{\appc}-subset $\sigma$ in $U$}
{\If{$\exists \sigma^*~\text{in}~Q_R$ with $reg(\sigma^*)\subseteq \sigma$ and replacing $\sigma^*$
by $\sigma$ does not make $Q_R$ invalid}{$Q_R=(Q_R\ominus \{\sigma^*\}\uplus\{\sigma\})$\;}}
\If{$Q_R$ contains no partial set}{return {\yes}\;}
\If{no $\sigma^*$ was replaced by some $\sigma$ in the above {\bf{foreach}} loop}{terminate the current procedure Extend()\;}
pick an arbitrary partial set $\sigma^*$ in $Q_R$\;
\ForEach{$c$ in $\mathbb{S}(Q_R)$}{\If{replacing a ``$*$'' in $\sigma^*$ by $c$ gives a partial set which is consistent with at least one set in $U$ and $T_R\ominus \{\sigma^*\}\uplus \{\sigma^*\ominus \{*\}\uplus \{c\}\}$ is valid}{$T_R=T_R\ominus \{\sigma^*\}\uplus \{\sigma^*\ominus \{*\}\uplus \{c\}\}$\;Extend($T_R$)\;}}
\caption{Extend($T_R$)}\label{algo:m3sp}
\end{procedure}

The algorithm correctly solves the instance since it enumerates all potential solutions. To see that the algorithm is an {\fpt}-algorithm, one can consider the algorithm as a bounded search tree algorithm, where each extension of a $T_R$ in the algorithm corresponds to a branching case. Since the size of $\mathbb{S}(Q_R)$ is bounded by $\mathappc\cdot R$, each node in the search tree has at most $\mathappc\cdot R$ children. Moreover, since each extension replaces one $*$ in $T_R$ with an element in $X$ and there are at most $(\mathappc-1)\cdot R$ many $*$'s, the depth of the search tree is bounded by $(\mathappc-1)\cdot R$. Finally, as there are at most $(r\cdot R)^R$ many $T_R$ to extend, the whole algorithm terminates in $O^*((\mathappc\cdot R)^{(\mathappc-1)\cdot R})\cdot O((r\cdot R)^R)=O^*((r\cdot R)^{r\cdot R})$ time.

Now let's turn our attention back to {\appc}-AV. To solve {\appc}-AV in {\fpt}-time, we develop a polynomial-time reduction from {\appc}-AV to M{\appc}SP. Let $\mathcal{E}=((\mathcal{C}, \Pi_{\mathcal{V}}), p \in \mathcal{C}, \Pi_{\mathcal{T}}, R)$ be an instance of {\appc}-AV.
Let $\Pi_p\sqsubseteq \Pi_{\mathcal{T}}$  be the multiset of all votes approving $p$ in $\Pi_{\mathcal{T}}$ and $\Pi_{\bar{p}}=\Pi_{\mathcal{T}}\ominus \Pi_p$.
The following observations are clearly true.

\begin{observation}\label{rbi}
If $R>|\Pi_p|$, then $\mathcal{E}$ is a {\yesins} if and only if $p$ wins the election $(\mathcal{C}, \Pi_{\mathcal{V}}\uplus \Pi_p)$.
\end{observation}

Due to the above observation, we can solve the problem in the case that $R>|\Pi_p|$: if $\Pi_p$ is a solution return ``Yes'', otherwise, return ``No''.

\begin{observation}\label{obsi3}
No solution of a {\yesins} contains an unregistered vote approving a candidate $c\in \mathcal{C}\setminus \{p\}$ such that $SC_{\mathcal{V}}(c)\geq SC_{\mathcal{V}}(p)+R-1$.
\end{observation}

\begin{observation}\label{rle}
If $R\leq |\Pi_p|$ and $\mathcal{E}$ is a {\yesins}, then there must be a solution containing exactly $R$ votes from $\Pi_p$ but none from $\Pi_{\bar{p}}$.
\end{observation}

Let $C'=\{c\in \mathcal{C}\setminus \{p\} \mid SC_{\mathcal{V}}(c)\geq SC_{\mathcal{V}}(p)+R-1\}$. Due to Observation~\ref{obsi3}, we can safely remove all votes that approve some candidate in $C'$ in $\Pi_{\mathcal{T}}$. Based on Observation~\ref{rle}, we reduce $\mathcal{E}$ to an M{\appc}SP instance as follows: the universal set is $X=(\mathcal{C}\setminus \{p\})\setminus C'$; the mapping $f$ is defined as
$f(c)=SC_{\mathcal{V}}(p)+R-SC_{\mathcal{V}}(c)-1$ for every $c\in X$. Moreover, for every subset $A\subseteq \mathcal{C}\setminus C'$ such that there are exactly $\ell$ unregistered votes $\pi_v$ such that $1(v)\cap C'=\emptyset$ and $A\cup \{p\}=1(v)$, there are $\ell$ copies of $A$ in the multiset $U$.
\end{proof}

Finally, we arrive at control by deleting candidates.
Faliszewski, Hemaspaandra and Hemaspaandra~\cite{DBLP:journals/ai/FaliszewskiHH14} proved that control by deleting candidates for 1-Approval is {\npc} when restricted to Swoon-SP elections. Since Swoop-SP elections are a special case of 2-peaked elections, 1-DC-{\kpeak} where $\mathkpeak\geq 2$ is {\npc}. We strengthen the {\npcns} of 1-DC-3 by proving that the problem is {\wah} with the number of deleted candidates as the parameter.

\begin{theorem}\label{theorem:deletingcandidates}
1-DC-3 is {\wah} with respect to the number of deleted candidates.
\end{theorem}
\begin{proof}
We prove the theorem by an {\fpt}-reduction from the {\sc{Independent Set}} problem which is {\wah}~\cite{fellows99}.
For a linear order $\vec{A}=(a_1,a_2, \dots ,a_n)$ over $A=\{a_1,a_2, \dots ,a_n\}$ and a subset $B\subseteq A$, denote by $\vec{A}\setminus B$ the linear order obtained from $\vec{A}$ by deleting all elements in $B$. 
For an instance $\mathcal{E}=(G=(V,E),\kappa)$ of the {\sc{Independent Set}} problem we construct an instance $\mathcal{E}'$ of 1-DC-3 as follows.

{\bf{Candidates:}} $V\cup \{p,a,a_1,a_2, \dots ,a_{\kappa},b,b_1,b_2, \dots ,b_{\kappa}\}$, where $\{p,a,a_1, \dots ,a_{\kappa},b,b_1, \dots ,b_{\kappa}\}$ is disjoint from $V$.

{\bf{3-Harmonious Order:}} Let $\mathcal{F}=(c_1,c_2, \dots ,c_n)$ be an (arbitrary) order of $V$. Then, the 3-harmonious order $\mathcal{L}$ is given by $(b_{\kappa},b_{\kappa-1}, \dots ,b_1,b,p,a,a_1,a_2, \dots ,a_{\kappa},c_1,c_2, \dots,c_n)$.

{\bf{Votes:}} Let $m=|E|$. There are 7 submultisets of votes.
\begin{enumerate}
\item $2m-1$ votes defined as $(\mathcal{L}[a,c_n],\mathcal{L}[p,b_{\kappa}])$; 
\item  $2m$ votes defined as $(\mathcal{L}[p,c_n], \mathcal{L}[b,b_{\kappa}])$; 
\item $2m+\kappa-1$ votes defined as $(\mathcal{L}[b,b_{\kappa}],\mathcal{L}[p,c_n])$; 
\item for each edge $(c_i,c_j)\in E(G)$ with $i<j$, create one vote defined as $(c_i,c_j,\mathcal{L}[a,a_{\kappa}],\mathcal{L}[p,b_{\kappa}],\mathcal{F}\setminus \{c_i,c_j\})$; 
\item for each vertex $c_i$, create one vote defined as $(c_i,\mathcal{L}[p,a_{\kappa}],\mathcal{L}[b,b_{\kappa}],\mathcal{F}\setminus \{c_i\})$ 
and one vote defined as $(c_i,\mathcal{L}[a,a_{\kappa}],\mathcal{L}[p,b_{\kappa}],\mathcal{F}\setminus\{c_i\})$; 
\item $\kappa+1$ votes defined as $(\mathcal{L}[a_1,c_n],\mathcal{L}[a,b_{\kappa}])$; 
\item one vote defined as $(\mathcal{L}[b_1,b_{\kappa}],\mathcal{L}[b,c_n])$.  
\end{enumerate}
It is easy to verify that all votes are 3-peaked with respect to $\mathcal{L}$.

{\bf{Number of Deleted Candidates:}} $R=\kappa$.

$(\Leftarrow:)$ It is easy to verify that $\mathcal{E}$ is a {\yesins} implies $\mathcal{E}'$ is a {\yesins}: for every independent set $S$ of size $\kappa$, deleting all candidates in $S$ from the election  makes $p$ the unique winner.

$(\Rightarrow:)$ Suppose that $\mathcal{E}'$ is a {\yesins} and $S'$ is a solution with $|S'|\leq \kappa$. Observe that $b\not\in S'$, since otherwise all candidates in $\{b_1,b_2, \dots ,b_{\kappa}\}$ must be deleted to make $p$ the unique winner, implying that $|S'|\geq |\{b,b_1, \dots ,b_{\kappa}\}|= \kappa+1$, a contradiction. The same argument applies to the candidate $a$. However, in order to beat $b$, exactly $\kappa$ candidates from $V$ must be deleted so that $p$ can get extra $\kappa$ points from the constructed votes in Case~5. Since $|S'|\leq \kappa$, $S'$ must be a subset of $V$. Moreover, no two candidates in $S'$ are adjacent in the graph $G$, since otherwise the candidate $a$ would get at least one extra point from the constructed votes in Case~4, and $p$ cannot be the unique winner. Thus, $S'$ is an independent set of $G$.
\end{proof}

\section{Conclusion}
 {\kpeak}-peaked elections are a natural generalization of single-peaked elections where each vote is allowed to have at most {\kpeak} peaks with respect to an order over the candidates. Hence, 1-peaked elections are exactly single-peaked elections and ${\frac{\lceil m\rceil}{2}}$-peaked elections are general elections, where $m$ is the number of candidates. We first studied the complexity of control by adding/deleting votes/candidates for $r$-Approval in {\kpeak}-peaked elections and achieved both polynomial-time solvability results and {\npcns} results. Our study leads to many dichotomy results for the problems considered in this paper, with respect to the values of {\kpeak} and $r$. Then, we presented some results concerning the parameterized complexity of these problems in general elections as well as in {\kpeak}-peaked elections, with respect to the solution size. All our results hold for both the unique-winner model and the nonunique-winner model. In addition, our {\npcns} results apply to Approval and SP-AV as well. See Table~\ref{tableresults} for a summary.

In an attempt to derive an {\npcns} result and an fixed-parameter tractability result, we obtained two byproducts which are of independent interest. First, we proved that every graph of maximum degree 3 admits a 2-interval representation such that every 2-interval contains a trivial interval and, moreover, every two 2-intervals may only intersect at their endpoints. Second, we presented an {\fpt}-algorithm for a generalized {\appc}-{\sc{Set Packing}} problem, where each element in the given universal set is allowed to occur in more than one {\appc}-subset in the solution.

A direction for future research would be to study various control problems for further voting systems in {\kpeak}-peaked elections and other generalizations of single-peaked elections such as $k$-dimensional single-peaked elections, $k$-swap single-peaked elections, etc. Very recently, Yang~\cite{AAMAS17YangBordaSinlgePeaked} proved that control by adding/deleting votes for Borda in {\kpeak}-peaked elections is {\npc} even when {\kpeak}$=1$ (i.e., single-peaked elections). 

\section*{Acknowledgments}
We thank the anonymous reviewers of JCSS, AAMAS 2014 and M-PREF 2013 for their constructive comments.

\bibliographystyle{plain}


\newpage
\section*{Appendix}
\def\x{0.75}

{\bf{Proof of Lemma~\ref{lemma:ids}}}

To prove Lemma \ref{lemma:ids}, we prove the following claim first. For a 2-interval representation $\mathcal{I}(G)$ of a 2-interval graph $G$,
and a vertex $u\in V(G)$, we call intervals
in $\mathcal{I}_u(G)$ {\it{$u$-intervals}}. An interval $I$ is a {\it{left-most}} (resp. {\it{right-most}})
interval if there is no other interval $I'$ such that $l(I')<l(I)$ (resp. $r(I')>r(I)$).
For $S\subseteq V(G)$, let $G\setminus S$ be the graph obtained from $G$ by removing all vertices in $S$.
For two graphs $G=(V,E)$ and $G'=(V',E')$, let $G\cup G'$ be the graph $(V\cup V', E\cup E')$.

A {\it{path}} between two vertices $v_1,v_2$ in a graph is a sequence $(v_1, \dots ,v_t)$ of distinct vertices such
that there is an edge between every two consecutive vertices in the graph. A graph is {\it{connected}} if there is a path
between every pair of vertices. In addition, a graph with one vertex is also considered as connected.
A {\it{component}} of a graph $G$ is a maximal
connected induced subgraph of $G$.
A {\it{cycle}} is a sequence $(v_1, \dots ,v_t)$ of distinct vertices such that $(v_1, \dots ,v_t)$ is a path and there is an edge between $v_1$ and $v_t$.
\medskip

\begin{figure}
\begin{center}
\includegraphics[width=\textwidth]{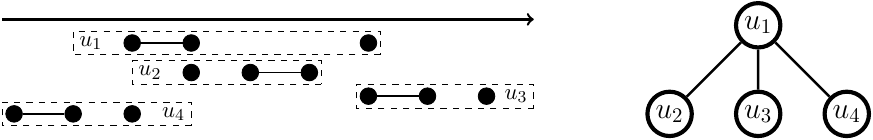}
\caption{This figure shows a 2-interval representation of the graph shown in {\myfig{figgraphproperty}}, obtained by turning over the 2-interval representation shown in
{\myfig{figgraphproperty}}. The left-most interval in the 2-interval representation shown in {\myfig{figgraphproperty}} is a $u_3$-interval.
In this figure, this $u_3$-interval is the right-most interval.}
\label{fig:intervalturnningover}
\end{center}
\end{figure}

\noindent\textbf{Claim.}
Let $G=(V,E)$ be a graph of maximum degree 3 and $v$ be any arbitrary vertex in $G$. Then, there is a 2-interval representation $\mathcal{I}(G,v)$ of $G$ such that the following four conditions hold.
\begin{enumerate}
\item the left-most intervals contain a $v$-interval;
\item each degree-($\leq 1$) vertex is represented by a trivial interval;
\item each degree-2 vertex $w$ is either represented by two trivial intervals, or by a closed interval $[x_1, x_2]$ with $x_1<x_2$ such that
the open interval $(x_1, x_2)$ does not intersect any $u$-interval for all $u\in V(G)\setminus \{w\}$;
\item for each degree-3 vertex $w\in V$, $\mathcal{I}_w(G,v)=\{I^1_w, I^2_w\}$ is one of the following two cases:
   \begin{enumerate}
      \item $I_w^1=[x_1, x_1]$, $I_w^2=[x_2, x_3]$, $x_1<x_2<x_3$ and
            $\nexists u\in V(G)\setminus \{w\}$ such that $r(I_u)\in (x_2, x_3)~\text{or}~l(I_u)\in (x_2, x_3)$;
      \item $I_w^1=[x_1, x_2]$, $I_w^2=[x_3, x_3]$, $x_1<x_2<x_3$ and
            $\nexists u\in V(G)\setminus \{w\}$ such that $r(I_u)\in (x_1, x_2)~\text{or}~l(I_u)\in (x_1, x_2)$,
    \end{enumerate}where $\mathcal{I}_w(G,v)$ is the 2-interval corresponding to $w$ in $\mathcal{I}(G,v)$ and $I_u$ is a $u$-interval in $\mathcal{I}(G,v)$.
\end{enumerate}

\begin{proof}
We prove the claim by induction on $n=|V(G)|$. The claim is clearly true for $n=1$. Notice that due to symmetry,
if there is a 2-interval representation where the left-most intervals contain a $v$-interval, then we can obtain a 2-interval representation
where the right-most intervals contain a $v$-interval by turning over the original 2-interval representation.
See~{\myfig{fig:intervalturnningover}} for an illustration.
Suppose that the claim is true for all graphs of maximum degree 3 with less than $n$ vertices.
Consider $G$ to be a graph with $n$ vertices.
Let $v$ be a vertex in $G$ as stated in the claim. 

Consider first the case that $v$ is a degree-3 vertex in $G$. Let $N_G(v)=\{a,b,c\}$.
We distinguish between the following cases.
\smallskip

\textbf{Case 1.} $v$ does not belong to any cycle in $G$.

Let $C_a, C_b$ and $C_c$ be the three components containing
$a, b$ and $c$ in $G\setminus \{v\}$, respectively.
By induction, we know that $C_a, C_b$ and $C_c$ have 2-interval
representations $\mathcal{I}(C_a,a), \mathcal{I}(C_b,b)$
and $\mathcal{I}(C_c,c)$, respectively, which satisfy all conditions stated in the above claim.
Moreover, let $\bar{G}=G\setminus (\{v\}\cup V(C_a\cup C_b\cup C_c))$ and * be any arbitrary vertex in $\bar{G}$ if there are any.
Let $\mathcal{I}(\bar{G},*)$ be a 2-interval representation of $\bar{G}$ satisfying all conditions stated in the claim (if $\bar{G}$ has no vertex, we ignore $\mathcal{I}(\bar{G},*)$ in all figures below).

In the following, we show how to construct a 2-interval representation $\mathcal{I}(G, v)$ satisfying all conditions stated in the above
claim by creating intervals for $v,a,b,c$ and combining them with $\mathcal{I}(C_a,a), \mathcal{I}(C_b,b), \mathcal{I}(C_c,c)$
and $\mathcal{I}(\bar{G},*)$, in a way as shown in the following figures
(in all figures, we ignore the real line for the sake of simplicity). Notice that in some cases of the construction, we turn over the 2-interval representations of some components of $G$.
We further distinguish between the following subcases mainly based on the shapes of the left-most intervals of $\mathcal{I}(C_a,a), \mathcal{I}(C_b,b)$
and $\mathcal{I}(C_c,c)$. Due to symmetry, it suffices to consider the following subcases.
\smallskip

{\textbf{Case 1.1.}} For each $x\in \{a,b,c\}$, the left-most $x$-interval in $\mathcal{I}(C_x,x)$ is a non-trivial interval. Clearly, in this case $a,b,c$ are degree-3 vertices in $G$. We can then construct $\mathcal{I}(G,v)$ as shown in the following figure.

\begin{center}
\includegraphics[width=\x\textwidth]{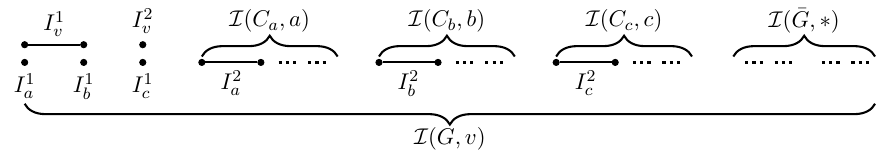}
\end{center}
\medskip

For ease of exposition, in all following cases, we only give the preconditions of each case and show $\mathcal{I}(G,v)$ in a figure, which is a refined combination of
the 2-interval representations of some components of $G$.
\medskip

{\textbf{Case 1.2.}} The left-most $b$-interval and the left-most $c$-internal in $\mathcal{I}(C_b,b)$ and $\mathcal{I}(C_c,c)$, respectively,
are non-trivial intervals, and the left-most $a$-interval in $\mathcal{I}(C_a,a)$ is a trivial interval. Clearly, both $b$ and $c$ are degree-3 vertices in $G$.
Assume first that $a$ is a degree-2 or degree-3 vertex in $G$.

\begin{center}
\includegraphics[width=\x\textwidth]{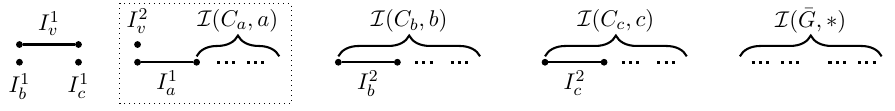}
\end{center}

If $a$ is a degree-1 vertex in $G$, according to the above claim, $a$ can only be represented by a trivial interval in $\mathcal{I}(G, v)$. In this case,
we replace the part in the rectangle in the above figure with the following one.

\begin{center}
\includegraphics[width=\x\textwidth]{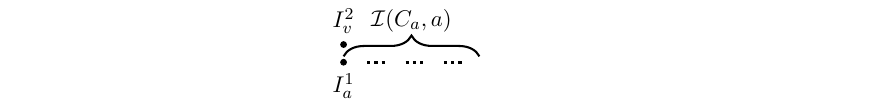}
\end{center}

{\textbf{Case 1.3.}} The left-most $c$-interval  in $\mathcal{I}(C_c,c)$ is a non-trivial interval, and the left-most $a$-interval and left-most $b$-interval in  $\mathcal{I}(C_a,a)$ and
$\mathcal{I}(C_b,b)$, respectively, are trivial intervals. Assume first that $a$ and $b$ are degree-2 or degree-3 vertices in $G$.
In the following figure, $\mathcal{I}(C_a,a)$ is turned over.

\begin{center}
\includegraphics[width=\x\textwidth]{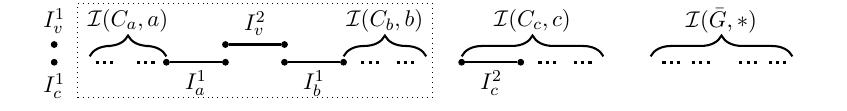}
\end{center}

In the case that $a$ or $b$, or both are degree-1vertices in $G$, we replace the part in the rectangle in the above figure with one of the following ones
(from left to right the figures correspond to the cases (1) both $a$ and $b$ are degree-1 vertices;
(2) only $a$ is a degree-1 vertex; (3) only $b$ is a degree-1 vertex).
\begin{center}
\begin{minipage}{0.33\textwidth}
\begin{center}
\includegraphics[width=0.8\textwidth]{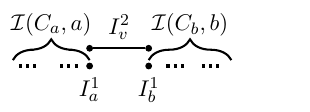}
\end{center}
\end{minipage}\begin{minipage}{0.33\textwidth}
\includegraphics[width=0.8\textwidth]{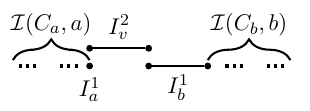}
\end{minipage}\begin{minipage}{0.33\textwidth}
\begin{center}
\includegraphics[width=0.8\textwidth]{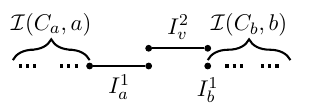}
\end{center}
\end{minipage}
\end{center}
\medskip

{\textbf{Case 1.4.}} For each $x\in \{a,b,c\}$, the left-most $x$-interval in $\mathcal{I}(C_x,x)$ is a trivial interval.
The following figure is for the case that $a,b,c$ are degree-2 or degree-3 vertices in $G$.
The case that some of $a,b,c$ or all of them are degree-1 vertices in $G$ are dealt with similar to Cases 1.2 and 1.3.

\begin{center}
\includegraphics[width=\x\textwidth]{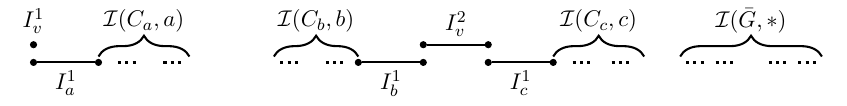}
\end{center}
\smallskip

\textbf{Case 2.} $v$ is in a cycle in $G$.

Obviously, two of $\{a,b,c\}$ are the neighbors of $v$ in the cycle. Without loss of generality, assume that $a$ and $b$ are the two neighbors of $v$
in the cycle denoted by $O=(v, a, u_1, u_2, \dots, u_t, b)$. We can construct $\mathcal{I}(O,v)$ as follows.

\begin{center}
\includegraphics[width=\x\textwidth]{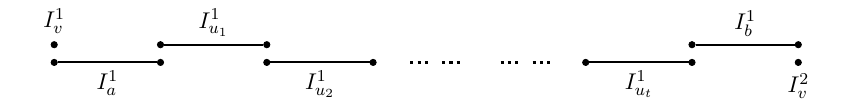}
\end{center}

Let $G'$ be the graph obtained from $G$ by deleting $v$ and all edges in the cycle $O$.
By the induction, there is a 2-interval representation $\mathcal{I}(G',c)$ of $G'$ which satisfies all conditions stated in the claim.
In particular, each vertex in $\{a, u_1, \dots ,u_t, b\}$
is represented by a trivial interval in $\mathcal{I}(G', c)$.
Then, by a similar method as used above we can construct $\mathcal{I}(G,v)$ by combining $\mathcal{I}(O,v)$ and $\mathcal{I}(G',c)$. We have three subcases here.
\smallskip

\textbf{Case 2.1.} The left-most $c$-interval in $\mathcal{I}(G', c)$ is a non-trivial interval. In this case, $c$ must be a degree-3 vertex in $G$.

\begin{center}
\includegraphics[width=\x\textwidth]{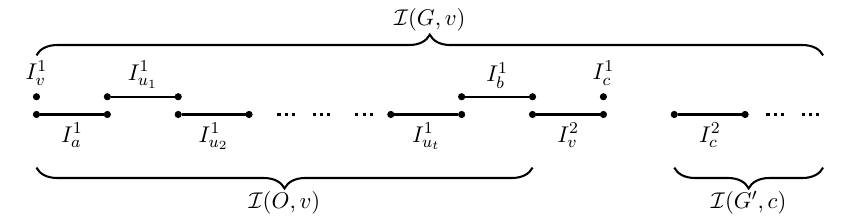}
\end{center}
\smallskip

\textbf{Case 2.2.} The left-most $c$-interval in $\mathcal{I}(G', c)$ is a trivial interval and $c$ is a degree-2 or degree-3 vertex in $G$.

\begin{center}
\includegraphics[width=\x\textwidth]{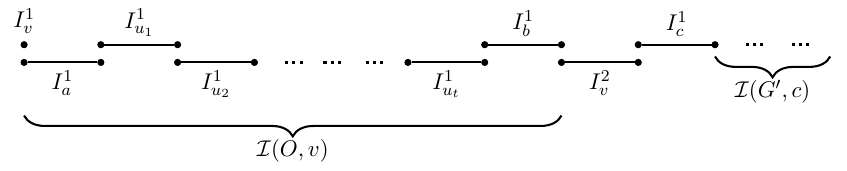}
\end{center}

\textbf{Case 2.3.} The left-most $c$-interval in $\mathcal{I}(G', c)$ is a trivial interval and $c$ is a degree-1 vertex in $G$.
\begin{center}
\includegraphics[width=0.85\textwidth]{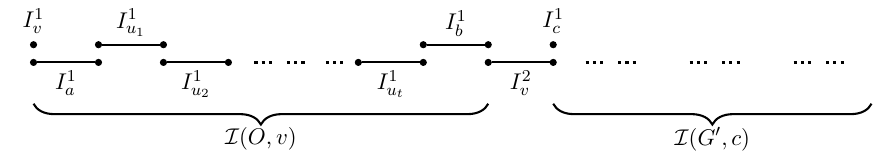}
\end{center}

Consider now the case that $v$ is a degree-2 vertex, say $N_G(v)=\{a,b\}$. The
proof can be much simpler. In particular, if $v$ is not in any cycle in $G$, we construct $\mathcal{I}(G, v)$ similar to Case 1 by distinguish between several subcases, each of which
can be obtained from one of Cases 1.1-1.3 by removing $\mathcal{I}(C_c,c)$ and $c$-intervals. In addition, if a trivial $v$-interval intersects a $c$-interval, we remove
this $v$-interval too. Moreover, if a non-trivial $v$-interval intersects a $c$-interval, we remove the edge of the $v$-interval and the common endpoint of the $v$-interval and the $c$-interval.
For instance, if the left-most $a$-interval in $\mathcal{I}(C_a, a)$ is a trivial interval, the left-most $b$-interval in $\mathcal{I}(C_b, b)$ is a non-trivial interval, and $a$ is a degree-2 vertex
in $G$, we construct $\mathcal{I}(G, v)$ as follows (obtained from Case 1.2 by performing the above operations). Here, $\mathcal{I}(C_x,x)$ where $x\in \{a,b\}$ is as defined in Case 1.

\begin{center}
\includegraphics[width=\x\textwidth]{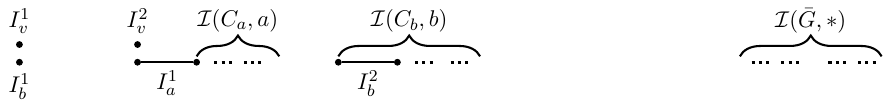}
\end{center}

If $v$ is in some cycle $O=(v, a, u_1, u_2, \dots, u_t, b)$, let $G'$ be obtained from $G$ by deleting $v$ and all edges in $O$. Let $c'$ be any vertex in $V\setminus \{v,a,b\}$.
Then, we can construct $\mathcal{I}(G, v)$ by combining $\mathcal{I}(O,v)$ and $\mathcal{I}(G',c')$ without any intersection,
where $\mathcal{I}(O,v)$ is as shown in Case 2 and $\mathcal{I}(G',c')$ is a 2-interval representation of $G'$ satisfying all conditions in the claim.

For the case that $v$ is a degree-1 vertex in $G$, let $N_G(v)=\{a\}$ and $G'=G\setminus \{v\}$. The following 3 figures show $\mathcal{I}(G,v)$ in all possible cases.

\begin{center}
\begin{minipage}{0.33\textwidth}
\begin{center}
\includegraphics[width=0.6\textwidth]{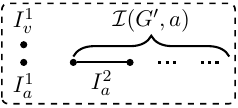}
\end{center}
\end{minipage}\begin{minipage}{0.33\textwidth}
\begin{center}
\includegraphics[width=0.6\textwidth]{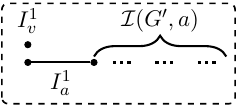}
\end{center}
\end{minipage}\begin{minipage}{0.33\textwidth}
\begin{center}
\includegraphics[width=0.6\textwidth]{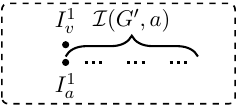}
\end{center}
\end{minipage}
\end{center}

Finally, if $v$ is an isolated vertex in $G$, we simply create a trivial interval for $v$ aside the 2-interval representation of $G\setminus \{v\}$ satisfying all conditions in the claim.

The above proof implies a polynomial-time algorithm for the construction of
a 2-interval representation that satisfies all conditions stated in the claim of a graph $G$
of maximum degree 3 and a vertex $v$ of $G$. In general, the algorithm iteratively decomposes the graph $G$ based on the above cases on $v$. For subgraphs of constant
sizes in the decompositions, we directly calculate the
2-interval representations satisfying all conditions in the claim in polynomial-time. Then, we iteratively combine the 2-interval representations using the constructions in
the above cases to construct
the desired 2-interval representation of $G$.

Given a 2-interval representation $\mathcal{I}(G,v)$ as discussed above, a 2-interval representation of $G$ satisfying all conditions stated in Lemma~\ref{lemma:ids}
can be computed in polynomial-time, by adding some dummy intervals for degree-($\leq 2$) vertices.
\end{proof} 
\end{document}